\documentclass[a4paper,11pt]{iopart}

\usepackage{fullpage}
\usepackage[T1]{fontenc}
\usepackage{graphicx,amsfonts,amsbsy,amssymb,amsthm}
\usepackage{url}
\usepackage{xcolor}

\newtheorem{theorem}{Theorem}
\newtheorem{dfn}{Definition}

\newcommand{\be}{\begin{equation}}
\newcommand{\ee}{\end{equation}}
\newcommand{\ba}{\begin{eqnarray}}
\newcommand{\ea}{\end{eqnarray}}
\newcommand{\ve}{\varepsilon}

\newcommand{\CB}{{\cal B}}
\newcommand{\CY}{{\cal Y}}

\newcommand{\CM}{{\cal M}}
\newcommand{\W}{ {\cal W}}
\newcommand{\E}{\mathbb{E}}

\renewcommand{\H}{\mathcal{H}}
\newcommand{\A}{\mathcal{A}}
\newcommand{\R}{\mathbb{R}}
\renewcommand{\O}{\mathcal{O}}

\newcommand{\floor}[1]{\lfloor #1 \rfloor}

\newcommand{\Av}[1]{\left\langle #1\right\rangle}
\newcommand{\Var}{{\rm Var}}
\newcommand{\Cov}{{\rm Cov}}

\begin{document}

\maketitle

\title{ Spectral statistics of the uni-modular ensemble}
\author{Christopher H. Joyner$^{1}$, Uzy Smilansky$^{2}$ and Hans A. Weidenm{\"u}ller$^{3}$}
\address {$^1$School of Mathematical Sciences, Queen Mary University of London, London, E1 4NS, UK}
\address {$^2$Department of Physics of Complex Systems,
Weizmann Institute of Science, Rehovot 7610001, Israel}
\address{$^3$Max-Planck-Institut f{\"u}r Kernphysik, Heidelberg , Germany}
\ead{\mailto{c.joyner@qmul.ac.uk}, \mailto{uzy.smilansky@weizmann.ac.il}, \mailto{haw@mpi-hd.mpg.de}}
\date{\today}

\begin{abstract}
We investigate the spectral statistics of Hermitian matrices in which
the elements are chosen uniformly from $U(1)$, called the uni-modular
ensemble (UME), in the limit of large matrix size. Using three
complimentary methods; a supersymmetric integration method, a
combinatorial graph-theoretical analysis and a Brownian motion
approach, we are able to derive expressions for $1 / N$ corrections to
the mean spectral moments and also analyse the fluctuations about this
mean. By addressing the same ensemble from three different point of
view, we can critically compare their relative advantages and derive
some new results.
\end{abstract}

\section{Introduction}
\label{intro}

In this work we investigate the spectral statistics of the uni-modular
ensemble of random matrices (UME) in the limit of large matrix
dimension $N$. The ensemble is defined as the set of Hermitian
matrices $\CM = \{ M_{\mu \nu} \}$, where $\mu, \nu = 1, 2, \ldots,
N$, with elements of the form
\be
\label{UME Definition}
M_{\mu\nu} = (1 - \delta_{\mu \nu}) \exp \{ i \phi_{\mu \nu} \} \ \
{\rm and } \ \ \phi_{\mu \nu} = - \phi_{\nu \mu},
\label{1}
\ee
and where, except for the symmetry relation, the phases $\phi_{\mu
  \nu}$ are uncorrelated real random variables distributed uniformly
in the interval $[0, 2 \pi)$. In contrast, the Gaussian unitary
  ensemble (GUE) is given by the set of Hermitian matrices $\H =
  \{H_{\mu\nu}\}$ of dimension $N$ endowed with the probability
  distribution
\begin{equation} \label{GUE definition}
P(\H) = \frac{1}{Z_N} \exp \left\{ - \frac{1}{2} {\rm Tr} (\H^2)
\right\}
\end{equation}
where $Z_N$ is a normalization factor. The UME serves as a
paradigmatic example of a Wigner
ensemble, i.e. a set of random Hermitian
matrices with independently distributed elements that do not follow a
Gaussian distribution (see e.g. \cite{Anderson-2009} for
details). Wigner~\cite{Wigner-1955,Wigner-1958} was the first to show
that many spectral properties of Wigner ensembles coincide with those
of the GUE in the limit of large $N$. Little is known, however, about
the $1/N$ deviations of particular Wigner ensembles from the universal
limit.

We address this question for the UME. We have chosen that ensemble
because, as we shall show, there exist exact relations that allow us
to determine the spectral properties of the UME even though the
ensemble is not unitarily invariant. That allows us to compute the
$1/N $ corrections to the universal limit. To our knowledge there has
been little previous work on uni-modular ensembles, except for the
work by Sodin and Feldheim, who investigate the fluctuations of
spectral moments in the UME and the minimum eigenvalue of unimodular
covariance matrices \cite{Sodin-2007,Sodin-2009,Sodin-2016}, and that
of Lakshminarayan, Puchala and Zyczkowski \cite{Lakshminarayan-2014},
who obtain exact expressions for the first four spectral moments of
unimodular covariance matrices\footnote{In contrast to the present
  article, the term `unimodular ensemble' is used in
  \cite{Lakshminarayan-2014} for the non-Hermitian counterpart of
  (\ref{UME Definition}).} and provide a conjecture for all moments.

The article is structured as follows: In the remainder of the present
Section we briefly discuss the distribution of the spectral moments in
the context of the GUE and related ensembles. We then use a general
argument to show that in the large $N$ limit the probability
distribution of these moments in the UME approaches that of the GUE
and discuss the difficulties in proceeding with corrections. In
Section~\ref{susy} we obtain corrections to the mean spectral density
via a supersymmetric approach and point to the difficulties of going
to higher order. In Section~\ref{Graph theory approach} we show how
these difficulties may be addressed in some cases using a
graph-theoretical and combinatorial techniques. Finally,
  in Section~\ref{BM approach} we use a Brownian motion approach,
  similar in spirit to Dyson's original approach \cite{Dyson-1962},
  with a theorem of Meckes \cite{Meckes-2009} utilising Stein's method, to show that the
  fluctuations of the spectral moments are Gaussian in the large $N$
  limit, and provide rates of convergence.

\subsection{Spectral moments in random-matrix theory}

The spectral moments $\tau_k(\H)$ of a Hermitian matrix $\H$ are given
by the moments of the empirical spectral density,
\begin{equation} \label{Spectral moments}
\tau_k(\H) = \int d \lambda \ \lambda^k \left( \frac{1}{N} \sum_{i=1}^N
\delta(\lambda - \lambda_i(\H)) \right) = \frac{1}{N} {\rm Tr} (\H^k).
\end{equation}
Here $\lambda_i$ with $i = 1, 2, \ldots, N$ are the eigenvalues of
$\H$. For Wigner matrices, the mean values of these moments
vanish for odd $k$ and, for even $k = 2 \nu$, converge in
  the limit of large matrix dimension to $C_\nu N^k$ where $C_\nu = (2
  \nu)! / (\nu! (\nu + 1)!)$ are the Catalan
  numbers~\cite{Wigner-1955,Wigner-1958}. In addition, the variances
converge sufficiently rapidly to conclude that the density converges
weakly, almost surely, to Wigner's semi-circle law
$\sigma(\lambda) = \frac{1}{(2\pi)}\sqrt{4 - \lambda^2}$
(see e.g. \cite{Anderson-2009} for instance). For the GUE, with the
average defined by (\ref{GUE definition}) and the mean values of the
moments (indicated by angular brackets) by

\begin{equation}\label{GUE Moment dfn}
 m_k := \Av{\tau_{2k}\left(\frac{\H}{\sqrt{N}}\right)} = \Av{\frac{1}{N^{k+1}}\Tr(\H^{2k})} \ .
 \end{equation}
Harer and Zagier~\cite{Harer-1986} discovered the three-term
recurrence relation
\begin{equation}\label{GUE recurrence}
(k+1)m_k = (4k-2)m_{k-1} + \frac{(k-1)(2k-1)(2k-3)}{N^2}m_{k-2} \ .
\end{equation}
That recurrence relation immediately leads to the following correction
to Wigner's leading term,
\begin{equation}\label{GUE moments}
m_k = C_k \left(1 + \frac{1}{N^2}\frac{(k-1)k(k+1)}{12} + \O(N^{-4})\right).
\end{equation}
Recurrence relations similar to (\ref{GUE recurrence}) were found for
the GOE and GSE by Ledoux \cite{Ledoux-2009} and for ensembles
characterized by the index $\beta$ by Witte and Forrester \cite
{Witte-2014}. For Gaussian, Laguerre and Jacobi $\beta$ ensembles,
exact expressions for the moments were given by Mezzadri, Reynolds and
Winn in terms of Jack polynomials~\cite{Mezzadri-2016}. These,
however, do not seem to lend themselves to asymptotic expansions in $1
/ N$. The systematic approach to the $1/N$ expansion has also been
addressed in the context of RMT distributions of the mean delay time.
These involve random matrix ensembles with exact expressions for the
joint probability density functions of the eigenvalues (see
e.g. \cite{Mezzadri-2012,Mezzadri-2013,Cunden-2016} and references
therein). We are not aware of attempts to go beyond leading order in
other matrix ensembles.

Fluctuations of the moments (for an arbitrary ensemble defined in
analogy to (\ref{Spectral moments})) are often discussed in terms of
the so-called \emph{linear-statistic}
\begin{equation}\label{Linear statistic}
L_f(\H):= {\rm Tr} [f(\H)] - \Av{{\rm Tr} [f(\H)]},
\end{equation}
where ${\rm Tr} [f(\H)] := \sum_{i=1}^N f(\lambda_i(\H))$. We note
that if $f$ is a polynomial then ${\rm Tr} [f(\H)]$ is simply a
weighted sum over the moments. The distribution of $L_f(\H)$ and
related quantities were first analysed by Jonson in the case of
Wishart matrices \cite{Jonson-1982}, by Johansson in the case of
unitarily invariant matrices \cite{Johansson-1998}, and by a number of
authors in the case of Wigner matrices
\cite{Khorunzhy-1996,Sinai-1998}. In all cases one observes
convergence to a Gaussian distribution, with a universal variance, in
the limit of large matrix size.

There exists a large number of papers - too many for detailed
referencing - that prove universality of $L_f(\H)$ for various types
of random matrices. We only emphasise some results that are
particularly relevant to this article. Using similar techniques to
those presented here, Sodin has shown for the UME that the moments of
$L_f(\CM)$ are Gaussian (see e.g. \cite{Sodin-2016} and references
therein) but does not discuss rates of convergence. Chatterjee has
previously used Stein's method along with estimations of Poincar\'{e}
inequalities to provide bounds on the total-variation distance between
a Gaussian and $L_f(\H)$ for appropriate random matrix ensembles
\cite{Chatterjee-2007}. Finally Cabanal-Duvillard has used a
Brownian-motion approach to derive similar results for the GUE
\cite{Cabanal-Duvillard-2001}, using the eigenvalue motion directly.

\subsection{Moments of the UME}
\label{moments}

We show that for $N \to \infty$, the distribution of the spectral
moments of the UME coincides with that of the GUE. We do so by showing
that in the limit, all moments and all products of moments of the UME
have the same values as for the GUE. The latter, defined in (\ref{GUE
  definition}), consists of Hermitean matrices ${\cal H}$ with
elements $H_{\mu \nu} = H^*_{\nu \mu}$ that are Gaussian-distributed
zero-centred random variables with second moments
\be
\Av{ H_{\mu \nu} H_{\nu' \mu'} } = \delta_{\mu \mu'}
\delta_{\nu \nu'} \ .
\label{2}
\ee
The normalization of the matrix elements in (\ref{2}) implies that
the support of the spectral density is $(- 2 \sqrt{N}, + 2
\sqrt{N})$. The elements of the UME have zero average and second
moments
\be
\langle M_{\mu \nu} M_{\nu' \mu'} \rangle = \delta_{\mu \mu'}
\delta_{\nu \nu'} (1 - \delta_{\mu \nu}) \ .
\label{3}
\ee
We note that for $\mu \neq \nu$, $| M_{\mu \nu} |^2 = 1$ without
averaging.

We first show that to leading order in $1 / N$ we have
\be
\langle {\rm Tr} ({\cal H}^n) \rangle = \langle {\rm Tr} ({\cal M}^n)
\rangle \ {\rm for \ all \ integer} \ n \geq 1 \ .
\label{4}
\ee
The Gaussian distribution of $H_{\mu \nu}$ implies $\langle {\rm Tr}
({\cal H}^n) \rangle = 0$ for odd values of $n$. For $n = 2 k$ even,
the trace is calculated using Wick contraction. Contributions of
leading order in $1 / N$ arise only from a subset of all Wick
contraction patterns (``nested contractions'') where contraction lines
connecting pairs of contracted matrix elements do not intersect. The
result is
\be
\langle {\rm Tr} ({\cal H}^{2 k}) \rangle = C_k N^{k + 1} + \ldots \ .
\label{5}
\ee
The Catalan numbers count the number of nested contractions in the
$n^{th}$ moment, with $n = 2 k$. The dots indicate terms of order
$N^l$ with $l \leq k$.

For the UME, we obviously have $\langle {\rm Tr} ({\cal M}^n) \rangle
= 0$ for $n$ odd. For ${\rm Tr}({\cal M}^{2})$, (\ref{3}) yields $N^2
- N$. That differs from the GUE result ${\rm Tr}({\cal H}^2) = N^2$ by
a term of order $N^{- 1}$. That term is due to the last Kronecker
symbol in (\ref{3}). Higher even moments of ${\cal M}$ receive
contributions not only from the pairwise correlations displayed in
(\ref{3}), but also from correlations of order $4, 6, ...$. We
demonstrate the existence of such correlations for the case of order
four.  The average of the term $M_{\mu \nu} M_{\mu' \nu'} M_{\mu''
  \nu''} M_{\mu''' \nu'''}$ vanishes unless the indices are pairwise
equal but differ within each factor ${\cal M}$. For the correlation of
order four (as opposed to a product of correlations of order two) that
gives
\ba
&& \langle M_{\mu \nu} M_{\mu' \nu'} M_{\mu'' \nu''} M_{\mu''' \nu'''}
\rangle_4 = [\delta_{\mu \nu'} \delta_{\nu \mu'} (1 - \delta_{\mu \nu})]
\nonumber \\
&& \qquad \times [\delta_{\mu'' \nu'''} \delta_{\nu'' \mu'''} (1 -
\delta_{\mu'' \nu''})] \delta_{\mu \mu''} \delta_{\nu \nu''} + \ldots \ .
\label{6}
\ea
The factors in straight brackets impose the conditions $\phi_{\mu \nu}
= - \phi_{\mu' \nu'}$ and $\phi_{\mu'' \nu''} = - \phi_{\mu'''
  \nu'''}$. The last two Kronecker deltas yield $\phi_{\mu \nu} =
\phi_{\mu'' \nu''}$, a condition that would be absent for a product of
two correlations of order two. The dots indicate terms obtained by a
permutation of the indices. Direct calculation yields ${\rm Tr}({\cal
  M}^4) = 2 N^3 - 3 N^2 + N$. That differs from the GUE result ${\rm
  Tr}({\cal H}^4) = 2 N^3 + N$ by terms of order $1 / N$. The
difference is due to the last Kronecker delta in (\ref{3}) and to the
two Kronecker deltas in (\ref{6}). Each Kronecker delta reduces the
number of independent summations in the expression for the trace and,
thus, produces terms of order $1 / N$. Correlations of higher order
than four exist and carry additional Kronecker symbols beyond the ones
in (\ref{6}). Therefore, with increasing $n$ the expressions for ${\rm
  Tr}({\cal M}^{2 n})$ and for ${\rm Tr}({\cal H}^{2 n})$ become ever
more different. The differences are confined, however, to terms of
order $1 / N$ or smaller. The term of leading order in ${\rm Tr}({\cal
  M}^{2n})$ is obtained by taking account only of binary correlations
and by omitting the last Kronecker delta in (\ref{3}). Hence,
\be
\langle {\rm Tr} ({\cal M}^{2 k}) \rangle = C_k N^{k + 1} + \ldots \ .
\label{7}
\ee
The terms indicated by dots are of lower order in $N$. They differ
from the terms indicated in the same manner in (\ref{5}).  Comparison
of (\ref{5}) and (\ref{7}) shows that all moments of the GUE and of
the UME become asymptotically ($N \to \infty$) equal, and that
(\ref{4}) holds.

To see in which sense (\ref{4}) applies we consider next-order
corrections. We start with corrections due to last Kronecker delta in
(\ref{3}). In (\ref{7}), terms of order $N^k$ arise from all nested
contributions involving that term once. The additional Kronecker delta
can be affixed to each one of the pairwise contractions ($k$
pairs). The number of contributions is, therefore, equal to $k C_k$
and the total contribution is given by $k C_k N^k$.  In comparison
with the result~(\ref{7}) that term is of order $k/N$.  For fixed $k$
the contribution vanishes for $N \to \infty$. It does not vanish,
however, for fixed $N$ and $k \to \infty$. An analogous conclusion
holds for the contribution of correlations of higher order to the
right-hand side of (\ref{7}). We conclude that (\ref{4}) is an
asymptotic relation. It establishes the identity of the $k^{\rm th}$
moments for fixed $k$ and $N \to \infty$.

We turn to the average of products of moments and show that for all
positive integer $k, n_1, n_2, \ldots, n_k$ and to leading order in
$N$ we have
\be
\langle {\rm Tr} ({\cal H}^{n_1}) \times \ldots \times {\rm Tr}
({\cal H}^{n_k}) \rangle = \langle {\rm Tr} ({\cal M}^{n_1}) \times
\ldots \times {\rm Tr} ({\cal M}^{n_k}) \rangle \ .
\label{8}
\ee
The left-hand side of (\ref{8}) is evaluated by calculating all Wick
contractions of pairs of matrix elements of ${\cal H}$. That rule
comprises pairs of matrix elements occurring under the same trace and
pairs that occur in different traces. Only nested contributions
contribute to the leading order in $N$. The right-hand side of
(\ref{8}) is evaluated using the binary correlation of (\ref{3}) as
well as all higher-order correlations as exemplified in (\ref{6}).
These likewise comprise sets of matrix elements that occur either
under the same trace or under two or more different traces. As in the
case of (\ref{4}) we use the fact that the leading-order terms in $N$
are obtained by suppressing the minimum number of independent
summations over matrix indices. That rules out all higher-order
correlations and leaves us with the binary correlations of
(\ref{3}). For the terms of leading order in $N$ we suppress the last
Kronecker delta in that equation. As a result we find that the
leading-order terms in $N$ of the right-hand side of (\ref{8}) are
obtained by calculating all Wick contractions of matrix elements of
${\cal M}$ (occurring either under the same trace or as arguments of
different traces). For each Wick-contracted pair the rule is the same
as for the GUE in (\ref{2}). The rules for calculating the right-hand
side of (\ref{8}) being the same as for the left-hand side, the
results are the same, too, and (\ref{8}) is seen to hold to leading
order in $N$. Again, that is an asymptotic result. It holds for fixed
$n_1, n_2, \ldots, n_k$ and $N \to \infty$.

To leading order in $1 / N$, these results imply the equality of the
mean spectral density of the GUE and the UME and also the convergence
in distribution of the $\tau_k$ and by extension of the $L_f(\CM)$ for
polynomial functions $f$. They do not, however, allow us to obtain
corrections to this density or rates of convergence for the
distributions. These aspects are explored in subsequent sections.

\subsection {Mean spectral density}

The empirical density $\rho(E) = \frac{1}{N}\sum_{i=1}^N \delta(E -
\lambda_i(\H))$ (see Eqn. (\ref{Spectral moments})), normalised so that
$\int dE \rho(E) = 1$, can be written in terms of the retarded Green's
function $G^{(r)}(E) = (E^+ - \CM)^{-1}$, where $E^+ = E +
i\varepsilon$ with $\varepsilon$ infinitesimal and positive, as
\be
\rho(E) = - \frac{1}{N\pi} \lim_{\varepsilon \to 0} \Im {\rm Tr} [G^{(r)}(E)]
\ .
\label{9}
\ee
We expand the retarded Green's function for the UME as
\be
{\rm Tr} [G^{(r)}(E)] = \sum_{n = 0}^\infty \frac{1}{(E^+)^{n + 1}}
{\rm Tr} [\CM^n] \ .
\label{10}
\ee
We first calculate $\langle G^{(r)}(E) \rangle$ to leading order in $1
/ N$ and then consider the sub-leading contributions. Using the
expansion~(\ref{10}) and taking into account only nested contributions
to the average, we obtain the Pastur equation \cite{Pastur-1972}
\be
\langle G^{(r, N)}(E) \rangle = \frac{1}{E} + \frac{1}{E} \langle
{\cal M} (\langle G^{(r, N)}(E) \rangle) {\cal M} \rangle \langle
G^{(r, N)}(E) \rangle \ .
\label{11}
\ee
The upper index $N$ stands for the leading-order term. We use the
binary correlator~(\ref{3}), suppress the last Kronecker delta, and
obtain
\be
\langle G^{(r, N)}(E) \rangle = \frac{1}{E} + \frac{1}{E} {\rm Tr}
(\langle G^{(r, N)}(E) \rangle) \ \langle G^{(r, N)}(E) \rangle \ .
\label{12}
\ee
We take the trace of (\ref{12}) and solve the resulting quadratic
equation for ${\rm Tr} (\langle G^{(r, N)}(E) \rangle)$. That gives
\be
{\rm Tr} (\langle G^{(r, N)}(E) \rangle) = \frac{E}{2} - i \sqrt{N}
\sqrt{1 - \frac{E^2}{4 N}} \ .
\label{13}
\ee
The range of the spectral density is $(- 2 \sqrt{N}, + 2 \sqrt{N})$.
For the full Green function we find
\be
\langle G^{(r, N)}(E) \rangle_{\mu \nu} = \frac{1}{N} \bigg( \frac{E}{2}
- i \sqrt{N} \sqrt{1 - \frac{E^2}{4 N}} \bigg) \delta_{\mu \nu} \ .
\label{14}
\ee
To leading order the spectral density is the same as for the GUE, as
expected.

Correction terms of order $1 / N$ to (\ref{12}) arise when either the
last Kronecker delta in (\ref{3}) or the fourfold
correlation~(\ref{6}) are taken into account once. Non-nested
contributions and higher-order correlations do not contribute to that
order. Using in (\ref{11}) the last Kronecker delta in (\ref{3}) we
obtain
\be
\delta G^{(r, bin)}_{\mu \nu}(E) = - \frac{1}{E} \langle G^{(r, N)}(E)
\rangle_{\mu \mu} \langle G^{(r, N)}(E) \rangle_{\mu \nu} \ .
\label{15}
\ee
The additional contribution in (\ref{11}) due to the fourfold
correlation term~(\ref{6}) is
\ba
\hspace{-15mm}
\delta G^{(r, four)}_{\mu \nu}(E) &=& \frac{1}{E} \langle G^{(r, N)}(E)
\rangle_{\mu \mu} \bigg( \sum_{\rho \neq \mu} (\langle G^{(r, N)}(E)
\rangle_{\rho \rho} )^2 \bigg)
 \langle G^{(r, N)}(E) \rangle_{\mu \nu} \ .
\label{16}
\ea 
Equation~(\ref{14}) shows that $G_{\mu \mu}(E)$ is of order $1 /
N$. Therefore, both contributions~(\ref{15}) and (\ref{16}) are of
order $1 / N$ compared to the leading contribution in (\ref{12}).
Adding the results~(\ref{15}) and (\ref{16}) we obtain as a $(1 /
N)$-correction to the spectral density of the UME a polynomial of
fourth order in $\langle G^{(r, N)}(E) \rangle$. That correction is
completely different from the $1 / N$ oscillations of the spectral
density displayed in later sections of the paper. The reason is that
the Pastur equation is valid only asymptotically. It is derived with
the help of the same asymptotic expansion as used for the averaged
traces in Eq.~(\ref{4}). That approach cannot be used for a systematic
evaluation of terms of next order in $1 / N$.

In Section \ref{susy} and Section \ref{Graph theory approach} we go
beyond leading order by using a supersymmetry approach first
developed in \cite{Kal02,Sha15} and a graph theoretic approach
adapted from $d$-regular graphs \cite{OrenI,OrenII}.

\section{Supersymmetry}
\label{susy}

Equation~(\ref{8}) suggests that all level correlation functions for the
UME coincide to leading order in $1 / N$ with those of the GUE.  The
argument goes as follows. The $(P, Q)$ level correlation function for
the UME is defined as
\ba
\hspace{-25mm}
&& \langle {\rm Tr} G^{(r)}(E + \ve_1) \times \ldots \times {\rm Tr}
G^{(r)}(E + \ve_P) \nonumber \\
&& \ \ \times {\rm Tr} G^{(a)}(E - \tilde{\ve}_1) \times \ldots \times
{\rm Tr} G^{(a)}(E - \tilde{\ve}_Q) \rangle \ .
\label{17}
\ea
Here $G^{(r)}(E)$ and $G^{(a)}(E)$ are the retarded and the advanced
Green functions for the UMA, respectively. The increments $\ve_p$, $p
= 1, \ldots, P$ and $\tilde{\ve}_q$, $q = 1, \ldots, Q$ are of the
order of the mean level spacing. The $(P, Q)$ level correlation
function for the GUE has the same form except for the replacement
${\cal M} \to {\cal H}$ in each of the Green's functions.

We use the expansion~(\ref{10}) for ${\rm Tr} G^{(r)}(E)$ and proceed
correspondingly for ${\rm Tr} G^{(a)}(E)$. Each term in the resulting
expansion of the correlation function~(\ref{17}) contains an ensemble
average over products of traces of powers of ${\cal M}$ that has the
form of the right-hand side of (\ref{8}). We proceed analogously for
the level correlator of the GUE, expanding the Green's functions in
powers of ${\cal H}$. Each term in the resulting series is obtained
from the corresponding term of the UME by the formal replacement
${\cal M} \to {\cal H}$. That same replacement converts the ensemble
average over products of traces of powers of ${\cal M}$ into the
ensemble average over products of traces of powers of ${\cal H}$. With
(\ref{8}) showing that these averages are equal to leading order in
$N$ we conclude that all $(P, Q)$ level correlationfunctions of the
UME coincide with those of the GUE in that order.

The argument lacks stringency, however. It is based upon a
perturbative expansion of the correlation functions. In contrast to
the spectral density, all correlation functions possess a zero mode.
The two-point function, for instance, has a zero mode at $\ve_1 = 0 =
\tilde{\ve}_1$ and thus, perturbatively, a singularity. That is why we
turn to the supersymmetry approach where the zero mode is treated
exactly.

The one-point function is written as
\be
{\rm Tr} \frac{1}{E^+ - {\cal M}} = \frac{1}{2} \frac{\partial}
{\partial j}{\cal G}(j) \bigg|_{j = 0} \ {\rm where} \ {\cal G}(j) =
\frac{\det(E^+ - {\cal M} + j)}{\det(E^+ - {\cal M} - j)} \ .
\label{18}
\ee
The generating function ${\cal G}(j)$ is written as a
superintegral. The $2 N$-dimensional supervector \[ \psi = (s_1, \ldots,
s_N, \chi_1, \ldots, \chi_N)^T \] contains the commuting complex
variables $s_k$ and the anticommuting variabless $\chi_k$, $k = 1,
\ldots, N$ with $\int \chi_k {\rm d} \chi_k = (2 \pi)^{- 1/2} = \int
\chi^*_k {\rm d} \chi^*_k$ for all $k$. The integration measure is the
flat Berezinian ${\rm d} (\psi^*, \psi) = \prod_k {\rm d} \Re(s_k)
{\rm d} \Im(s_k) {\rm d} \chi^*_k {\rm d} \chi_k$. In the $2
N$-dimensional superspace (the direct product of the $N$-dimensional
ordinary space with indices $k = 1, \ldots, N$ and the two-dimensional
superspace with indices $s = 1, 2$) we define
\be
{\cal C} = (E^+ 1^N - {\cal M}) 1^s - j \sigma_3 1^N \ .
\label{19}
\ee
Here $1^s$ and $\sigma_3$ are the unit matrix and the third Pauli spin
matrix, respecively, in two-dimensional superspace while $1^N$ is the
unit matrix in ordinary space. With $\tilde{\psi} = (\psi^*)^T$ we
have
\be
{\cal G}(j) = \int {\rm d} (\psi^*, \psi) \ \exp \{ (i/2) \tilde{\psi}
{\cal C} \psi \} \ .
\label{20}
\ee
The ensemble average of ${\cal G}$ is defined as an average over the
independent phases $\phi_{k l}$ with $k < l$.

To average ${\cal G}$ we calculate the expectation value of $\exp \{ -
(i/2) (\tilde{\psi} {\cal M} 1^s \psi) \}$. We first consider the
moments of $(\tilde{\psi} {\cal M} 1^s \psi)$. All odd moments vanish.
For the second moment we use (\ref{3}). The fourth moment is the sum
of the binary and the fourfold correlations given in (\ref{3}) and
(\ref{6}), respectively. Thus,
\ba
\langle (\tilde{\psi} {\cal M} 1^s \psi)^2 \rangle &=& \sum_{k \neq l}
(\tilde{\psi}_k \psi_l) (\tilde{\psi}_l \psi_k) \ , \nonumber \\
\langle (\tilde{\psi} {\cal M} 1^s \psi)^4 \rangle &=& 3 \bigg[
\sum_{k \neq l} (\tilde{\psi}_k \psi_l) (\tilde{\psi}_l \psi_k) \bigg]^2
+ \sum_{k \neq l} \bigg[ (\tilde{\psi}_k \psi_l) (\tilde{\psi}_l \psi_k)
\bigg]^2 \ .
\label{21}
\ea
That gives
\ba
\fl\bigg\langle \exp \{ - (i/2) (\tilde{\psi} {\cal M} 1^s \psi) \}
\bigg\rangle &=& \exp \bigg\{ - \frac{1}{8} \sum_{k \neq l}
(\tilde{\psi}_k \psi_l)(\tilde{\psi}_l \psi_k) + \frac{1}{3 \cdot 8^2}
\sum_{k \neq l} \bigg[ (\tilde{\psi}_k \psi_l) (\tilde{\psi}_l \psi_k)
\bigg]^2 \bigg\} \ .
\label{22}
\ea
The first two terms in the exponent strongly suggest how the series
continues although we have not checked that. Writing
\be
\sum_{k \neq l} (\tilde{\psi}_k \psi_l) (\tilde{\psi}_l \psi_k) =
\sum_{k, l} (\tilde{\psi}_k \psi_l) (\tilde{\psi}_l \psi_k) - \sum_k
(\tilde{\psi}_k \psi_k)^2 \ ,
\label{23}
\ee
we observe that the second term on the right-hand side of (\ref{23})
is of order $1 / N$ compared to the first one. The same is true of the
second term on the right-hand side of the second of (\ref{21}) in
comparison with the first one. And the same statement holds {\it a
  fortiori} for higher-order correlations. To leading order in $1 / N$
we, therefore, have
\be
\bigg\langle \exp \{ - (i/2) (\tilde{\psi} {\cal M} 1^s \psi) \}
\bigg\rangle \approx \exp \bigg\{ - \frac{1}{8} \sum_{k, l}
(\tilde{\psi}_k \psi_l)(\tilde{\psi}_l \psi_k) \bigg\} \ .
\label{24}
\ee
For the GUE we have correspondingly
\be
\bigg\langle \exp \{ - (i/2) (\tilde{\psi} {\cal H} 1^s \psi) \}
\bigg\rangle = \exp \bigg\{ - \frac{1}{8} \sum_{k, l} (\tilde{\psi}_k
\psi_l) (\tilde{\psi}_l \psi_k) \bigg\} \ .
\label{25}
\ee
The equality of the right-hand sides of (\ref{24}) and (\ref{25})
implies that to leading order in $1 / N$ the spectral densities of the
UME and of the GUE are identical.

We turn to the $(P, Q)$ level correlation function of the UME as
defined in (\ref{17}). We skip the construction of the generating
function for the correlation function~(\ref{17}) because it runs
completely parallel to that for the GUE given in Ref.~\cite{Plu15}.
Suffice it to say that the result is similar in form to (\ref{20}),
with the following differences. The vectors $\psi$ and $\tilde{\psi}$
and the matrix ${\cal C}$ now have dimension $2 N (P + Q)$, the matrix
${\cal C}$ contains the matrix ${\cal M}$ in block-diagonal form $(P +
Q)$ times, the scalar $j$ becomes a matrix of dimension $(P + Q)$, the
vector $\psi$ ($\tilde{\psi}$) is multiplied from the left (right) by
the matrix ${\bf L}^{1/2}$ where ${\bf L} = {\bf 1}$ in the retarded
and ${\bf L} = - {\bf 1}$ in the advanced sector, and the energy
increments $\ve_1, \ldots, \ve_P$ and $\tilde{\ve}_1, \ldots,
\tilde{\ve}_Q$ appear in the exponent. Evaluating the expectation
value of $\exp \{ - (i/2) \tilde{\psi} {\cal M} \psi \}$ by using
(\ref{21}) and dropping terms of higher order in $1 / N$ we obtain
exactly the same form for the averaged generating function as for the
GUE. That implies that to leading order in $1 / N$, all $(P, Q)$ level
correlation functions for the UME coincide with those of the GUE.

There are two possibilities to go beyond these results. First, the
average over the phases $\phi_{\mu \nu}$ can be done exactly using the
color-flavor transformation~\cite{Zir96}. For quantum graphs, that
transformation was used in Refs.~\cite{Gnu04, Gnu05, Plu15}. In the
present context the color-flavor transformation seems uneconomical,
however. The reason is seen by considering the generating
function~(\ref{20}) for the spectral density. In (\ref{20}) the
ensemble average has to be taken by integrating over the real phase
angles $\phi_{\mu \nu}$ with $\mu < \nu$ and $\mu, \nu = 1, \ldots,
N$. The color-flavor transformation performs these integrations at the
expense of introducing for every $\phi_{\mu \nu}$ a pair of
supermatrices $(Z_{\mu \nu}, \tilde{Z}_{\nu \mu})$ with $\mu <
\nu$. That increases the number of integration variables by a factor
eight. We have, therefore, not followed that possibility. The second,
more attractive possibility is to examine correction terms of order $1
/ N$ within the Hubbard-Stratonovich approximation to the
supersymmetry approach. We do that, confining ourselves to the
spectral density.

\subsection{Terms of Order $1 / N$}

The papers by Kalisch and Braak~\cite{Kal02} and by
Shamis~\cite{Sha15} show how corrections of order $1 / N$ to the
spectral density of the GUE can be worked out. Shamis makes heavy use
of the unitary invariance of the GOE. Such invariance is not shared by
the UME. That is why we follow the work of Kalisch and Braak. These
authors use the Hubbard-Stratonovich transformation, perform the
integration over the two remaining anticommuting variables exactly,
and then use the saddle-point approximation. We apply that method in
our more general context.

We approximate the ensemble average of $\exp \{ - (i/2) \tilde{\psi}
{\cal M} \psi \}$ by keeping only correction terms of order $1 /
N$. That is done by dropping in (\ref{22}) the last term and by
writing the second term as in (\ref{23}). Thus,
\ba
\langle {\cal G}(j) \rangle &\approx& \int {\rm d} (\psi^*, \psi) \
\exp \{ (i/2) \tilde{\psi} (E 1^s - j \sigma_3 ) \psi \} \nonumber \\
&& \times \exp \{ - \frac{1}{8} \sum_{k, l} (\tilde{\psi}_k \psi_l)
(\tilde{\psi}_l \psi_k) + \frac{1}{8} \sum_k (\tilde{\psi}_k \psi_k)^2
\} \ .
\label{26}
\ea
We eliminate the first quartic term by a single Hubbard-Stratonovich
transformation and the second one by $N$ such transformations, one for
each term $(\tilde{\psi}_k \psi_k)^2$. We define the two-by-two
supermatrices
\be
{\cal A} = (i/2) \sum_k \psi_k ) ( \tilde{\psi}_k \ , \ {\cal B}_k =
(1/2) \psi_k ) ( \tilde{\psi}_k
\label{27}
\ee
so that
\be
- \frac{1}{8} \sum_{k, l} (\tilde{\psi}_k \psi_l) (\tilde{\psi}_l
\psi_k) = \frac{1}{2} {\rm STr}_s ({\cal A}^2) \ , \ \frac{1}{8}
(\tilde{\psi}_k \psi_k)^2 = \frac{1}{2} {\rm Str}_s ({\cal B}^2_k) \ .
\label{28}
\ee
We use
\ba
\exp \{ \frac{1}{2} {\rm STr}_s ({\cal A}^2) \} &=& \int {\rm d}
\sigma \ \exp \{ - \frac{1}{2} {\rm STr}_s (\sigma^2) - {\rm STr}_s
(\sigma {\cal A}) \} \ , \nonumber \\
\exp \{ \frac{1}{2} {\rm STr} ({\cal B}_k^2) \} &=& \int {\rm d}
\tau_k \ \exp \{ - \frac{1}{2} {\rm STr} (\tau_k^2) - {\rm STr}
(\tau_k {\cal B}_k) \} \ .
\label{29}
\ea
The supermatrices $\sigma$ and $\tau_k$ all have dimension two. We
insert that into (\ref{27}) and carry out the Gaussian integrals over
the original integration variables. That gives
\ba
\langle {\cal G}(j) \rangle &=& \int {\rm d} \sigma \ \exp \{ -
\frac{1}{2} {\rm STr}_s (\sigma^2) \} \prod_{k = 1}^N \bigg\{ \int {\rm d}
\tau_k \ \exp \{ - \frac{1}{2} {\rm STr}_s (\tau_k^2) \} \nonumber \\
&& \qquad \times \exp \{ - \sum_k {\rm STr}_{s} \ln \big( E 1^s - \sigma
- i \tau_k - j \sigma_3 \big) \} \bigg\} \ .
\label{30}
\ea
The indices indicate that the supertraces extend only over the
superindices $s$. We remark in parentheses that (\ref{30}) shows the
limitations of the supersymmetry approach. For the GUE we deal with a
single supermatrix $\sigma$. The $1 / N$ correction to the GUE
introduces $N$ additional supermatrices $\tau_k$. Corrections of
higher order in $1 / N$ lead to ever more complex integrals, bringing
the method to its limit.

\subsubsection{GUE}
\label{gue}

We are guided by the GUE case. We consider only terms up to first
order in $j$ and indicate that fact by using the sign $\approx$ in
Eqs.~(\ref{G1}) through (\ref{G4}),
\ba
&& \langle {\cal G}(j) \rangle = \int {\rm d} \sigma \ \exp \{ -
\frac{1}{2} {\rm STr}_s (\sigma^2) \} \exp \{ - N {\rm STr}_{s}
\ln \big( E 1^s - \sigma - j \sigma_3 \big) \} \bigg\}
\nonumber \\
&\approx& \int {\rm d} \sigma \ \exp \bigg\{ - \frac{1}{2} {\rm STr}_s
(\sigma^2) - N {\rm STr}_{s} \ln \big( E 1^s - \sigma\big) \bigg\}
\nonumber \\
&& \qquad \times \bigg( 1 + j N {\rm STr} \bigg( \frac{1}{E 1^s -
\sigma} \sigma_3 \bigg) \bigg) \ .
\label{G1}
\ea
We write
\ba
\sigma = \left( \matrix{ s_B & \alpha \cr
                         \alpha^* & i s_F \cr} \right) \ .
\label{G2}
\ea
Here $s_B, s_F$ are real commuting and $\alpha, \alpha^*$
anticommuting variables. We define $a = E^+ - s_B$, $b = E - i s_F$
and carry out the integrals over the anticommuting variables. Then
\ba
&& \langle {\cal G}(j) \rangle \approx \int_{- \infty}^{+ \infty}
{\rm d} s_B \int_{- \infty}^{+ \infty} {\rm d} s_F \ \exp \bigg\{ -
\frac{1}{2} (s^2_B + s^2_F) \bigg\} \bigg( \frac{b}{a} \bigg)^N
\nonumber \\
&& \times \bigg\{ 1 - \frac{N}{a b} + j N \bigg( \bigg[ \frac{1}{a}
+ \frac{1}{b} \bigg] \bigg[ 1 - \frac{N}{a b} \bigg] + \frac{a -
b}{a^2 b^2} \bigg) \bigg\} \ .
\label{G3}
\ea
We rescale $E \to x = E / \sqrt{N}$ and with it $s_B \to q = s_B /
\sqrt{N}$, $s_F \to p = s_F / \sqrt{N}$, $a \to c = a / \sqrt{N} = x^+
- q$, $b \to d = b / \sqrt{N} = x - i p$, and $j \to j' = j \sqrt{N}$.
The last operation assures that the average level density is
normalized to unity. Choosing $j' = j / \sqrt{N}$ would yield an
average level density normalized to $N$. Then
\ba
&& \langle {\cal G}(j') \rangle \approx N \int_{- \infty}^{+ \infty}
{\rm d} q  \int_{- \infty}^{+ \infty} {\rm d} p \ \exp \bigg\{ -
\frac{N}{2} (q^2 + p^2) \bigg\} \bigg( \frac{d}{c} \bigg)^N
\nonumber \\
&& \times \bigg\{ 1 - \frac{1}{c d} + j' \bigg( \bigg[ \frac{1}{c} +
\frac{1}{d} \bigg] \bigg[ 1 - \frac{1}{c d} \bigg] + \frac{1}{N}
\frac{c - d}{c^2 d^2} \bigg) \bigg\} \ .
\label{G4}
\ea
Using Eq.~(\ref{18}) we have performed the integrals over $p$ and $q$
analytically. The resulting expression in Hermite polynomials agrees
for every $N$ with the standard result.

For the saddle-point approximation we define the effective action
\be
{\cal A} = \frac{N}{2} (p^2 + q^2) + N \log (x - i p) - N \log (x -
q) \ .
\label{G5}
\ee
It is the sum of the effective actions ${\cal A}_q$ and ${\cal A}_p$
for the variables $q$ and $p$. Therefore, the saddle points for $p$
and for $q$ are unrelated. The saddle points for $q$ are $q_0 = (x/2)
\pm i \sqrt{1 - x^2 / 4}$. Because of the singularity of the integrand
at $x^+$ we admit only the the solution in the lower half plane so
that
\be
q_0 = \frac{x}{2} - i \sqrt{1 - \frac{x^2}{4}} \ .
\label{G7}
\ee
As $x^+$ moves from $- 2^+$ to $+ 2^+$ the saddle point $q_0$ moves
from $-1$ on a semicircle into the lower half plane, reaches the value
$- i$ for $x = 0$, and continues on the semicircle to $+ 1$ for $x =
2$. Without crossing the singularity, we can shift the path of
integration for all values of $x$ with $|x| \leq 2$ so that it runs
parallel to the real axis and passes through $q_0$. We write $q = q_0
+ t_B$ and expand ${\cal A}_q$ around $q_0$ in powers of $t_B$ up to
second order,
\ba
&& {\cal A}_q = \frac{N}{2} (q_0 + t_B)^2 + N \log ( x^+ - q_0 - t_B)
\nonumber \\
&& \qquad \approx \frac{N}{2} q^2_0 + N \log (x^+ - q_0) \nonumber \\
&& \qquad \qquad + N \bigg( q_0 - \frac{1}{x - q_0} \bigg) t_B
+ \frac{N}{2} \bigg( 1 - \frac{1}{(x - q_0)^2} \bigg) t^2_B \ .
\label{G6}
\ea
The saddle points for $p$ are
\be
p^\pm_0 = - i \frac{x}{2} \pm \sqrt{1 - \frac{x^2}{4}} \ .
\label{G9}
\ee
As $x$ increases from $-2$ to zero, the two saddle points (that are
degenerate at $x = - 2$ with value $+ i$) move on semicircles in
opposite directions, reaching the real axis at $- 1$ and at $+ 1$ for
$x = 0$ and continue into the lower half plane, reaching the
degenerate point $- i$ at $x = + 2$. For all values of $|x| \leq 2$
the two saddles lie on a straight line parallel to the real axis. We
shift the path of integration so that it runs along that line. For
${\cal A}_p$ we write $p = p_0 + t_F$ and expand ${\cal A}_p$ in
powers of $t_F$ up to second order,
\ba
&& {\cal A}_p = \frac{N}{2} (p_0 + t_F)^2 - N \log ( x - i p_0 - i t_F)
\nonumber \\
&& \qquad \approx \frac{N}{2} p^2_0 - N \log (x - i p_0) + N \bigg( p_0 +
\frac{i}{x - i p_0} \bigg) t_F \nonumber \\
&& \qquad \qquad + \frac{N}{2} \bigg( 1 - \frac{1}{(x - i p_0)^2}
\bigg) t^2_F \ .
\label{G8}
\ea
We evaluate $\langle {\cal G}(j') \rangle$ at the two saddle points
defined by $i p_0 = q_0$ and $i p_0 = q^*_0$. We do so by expanding up
to and including terms of order $1 / N$. Upon integration over $t_B$
and $t_F$ that gives for $i p_0 = q_0$ the result $\langle {\cal
  G}(j')^- \rangle = 1 + 2 j' q_0$. Equation~(\ref{18}) and the fact
that the imaginary part of the retarded Green function equals $ - \pi
\delta(E - {\cal H})$ then yields for the spectral density
\be
\rho(x) = \frac{1}{\pi} \sqrt{1 - x^2/4} \ .
\label{G12}
\ee
That is the asymptotic expression ($N \to \infty$). For the first
saddle point, it holds up to and including terms of order $1 / N$.

For $i p_0 = q^*_0$ the leading-order contribution is of order $1 / N$
and given by
\ba
\langle {\cal G}(j')^+ \rangle &=& \frac{i}{N} \exp \{ - \frac{N}{2}
( q^2_0 - (q^*_0)^2) \} \bigg( \frac{q_0}{q^*_0} \bigg)^N \bigg\{
\frac{x}{4 - x^2} + j' \frac{1}{1 - x^2/4} \bigg\} \ .
\label{G15}
\ea
The $(1 / N)$-correction to the level density~(\ref{G12})
is~\cite{Kal02, Sha15}
\ba
\delta \rho(x) = \frac{1}{N \pi} \frac{1}{1 - x^2/4} \exp \{ i N x
\sqrt{1 - x^2/4} + 2 i N \arctan (- 2 \sqrt{1 - x^2/4} / x ) \} \ .
\label{G16}
\ea
Characteristic features are the rapid oscillations with frequency
$1/N$ and the singularities at the end points $x = \pm 2$ of the
spectrum. For finite $N$ the exact expression for the spectral density
is non-singular for all values of the energy. The singularities occur
only in the $1 / N$ expansion.

\subsection{UME}

We start from (\ref{30}). In the calculation of the normalization
integral ${\cal G}(0)$, we proceed as in Section~\ref{gue}, see also
the calculation of the source terms discussed below. The effective
action ${\cal A}$ is defined as the contribution of leading order in
$1 / N$ to the negative exponent of the integrand in ${\cal G}(0)$.
We mention without proof that the effective action turns out to be
equal to the effective action for the GUE in Eq.~(\ref{G5}). The
saddle points are the same. When we calculate the leading-order
contribution of the two saddle points to ${\cal G}(0)$ we find that
the result is identical to the GUE expression in (\ref{G4}).
Therefore, ${\cal G}^-(0) = 1$ and ${\cal G}^+(0) = 0$.

We turn to the source terms. For $\sigma$ we use the
parametrization~(\ref{G2}). For $\tau_k$ with $k = 1, 2, \ldots, N$ we
write
\ba
\tau_k = \left( \matrix{ t_{k B} & \gamma_k \cr
                        \gamma^*_k & i t_{k F} \cr} \right)  \ .
\label{32}
\ea
Here $t_{k B}, t_{k F}$ are real commuting and $\gamma_k, \gamma^*_k$
are anticommuting variables. Then
\be
{\rm STr} (\tau^2_k) = t^2_{k B} + t^2_{k F} + 2 \gamma_k \gamma^*_k \ .
\label{33}
\ee
We define
\be
a_k = (E - s_B - i t_{k B}) \ , \ b_k = (E - i s_F + t_{k F}) \ ,
\label{34}
\ee
In (\ref{30}) we keep terms of first oder in $j$. We have
\ba
&&\exp \bigg \{ - {\rm STr}_{s} \ln \big( E 1^s - \sigma - i \tau_k - j
\sigma_3 \big) \bigg\} \to j \ {\rm STr}_{s} ( \sigma_3 ( E 1^s -
\sigma - i \tau_k )^{- 1} ) \nonumber \\
&& \quad \times \exp \bigg\{ - {\rm STr}_{s} \ln \big( E 1^s
- \sigma - i \tau_k \big) \bigg\} \ .
\label{63}
\ea
We use the identity
\be
\left( \matrix{ a & \phi \cr
              \phi^* & b \cr} \right)^{- 1} = \left( \matrix{
      (1/a) (1 + \phi \phi^* / (a b)) & - \phi / (a b) \cr
          - \phi^* / (a b) & (1/b) (1 + \phi^* \phi / (a b) \cr}
\right)
\label{64}
\ee
which is valid for an arbitrary supermatrix of dimension two. The
prefactor in expression~(\ref{63}) becomes
\be
j \bigg( \frac{1}{a_k} + \frac{1}{b_k} + \frac{b_k - a_k}{a^2_k b^2_k}
(\alpha + i \gamma_k)(\alpha^* + i \gamma^*_k) \bigg) \ .
\label{65}
\ee
A change of integration variables shows that that factor is the same
for every value of $k$. Therefore
\ba
&& \frac{\partial}{\partial j} {\cal G}(j) \bigg|_{j = 0} = N \int
{\rm d} \sigma \ \exp \{ - \frac{1}{2} {\rm STr}_s (\sigma^2) \}
\nonumber \\
&& \qquad \prod_{k = 1}^N \bigg\{ \int {\rm d} \tau_k \ \exp \{ -
\frac{1}{2} {\rm STr}_s (\tau_k^2) \} \nonumber \\
&& \qquad \qquad \times \exp \{ - \sum_k {\rm STr}_{s} \ln \big(
E 1^s - \sigma - i \tau_k \} \bigg\} \nonumber \\
&& \qquad \times \bigg( \frac{1}{a_1} + \frac{1}{b_1} + \frac{b_1
- a_1}{a^2_1 b^2_1} (\alpha + i \gamma_1)(\alpha^* + i \gamma^*_1)
\bigg) \ .
\label{66}
\ea
We integrate explicitly over all anticommuting variables. That gives
\ba
&& \frac{\partial}{\partial j} {\cal G}(j) \bigg|_{j = 0} =
\frac{N}{2 \pi} \int {\rm d} s_B \int {\rm d} s_F \ \exp \{ -
\frac{1}{2} ( s^2_B + s^2_F \} \nonumber \\
&& \ \times \bigg(\frac{1}{2 \pi}\bigg)^N \prod_{k = 1}^N \int
{\rm d} t_{k B} {\rm d} t_{k F} \exp \{ - \frac{1}{2} (t^2_{k B} +
t^2_{k F}) \} \frac{b_k}{a_k} \prod_{l = 2}^N \bigg( 1 + \frac{1}
{a_l b_l} \bigg) \nonumber \\
&& \ \times \bigg[ \frac{1}{a_1} + \frac{1}{b_1} + 2 \frac{a_1 -
b_1}{a^3_1 b^3_1} - (N - 1) \bigg( \frac{1}{a_1} + \frac{1}{b_1} +
\frac{2}{a^2_1 b_1} \bigg) \frac{1}{1 + a_2 b_2} \bigg] \ .
\label{66a}
\ea
Integration over the commuting variables $\tau_{k, B}$ and $\tau_{k
  F}$ shows that the leading terms in powers of $1 / N$ are
\ba
&& \frac{\partial}{\partial j} \langle {\cal G}(j) \rangle
\bigg|_{j = 0}  = \frac{N}{2 \pi} \int_{- \infty}^{+ \infty} {\rm d}
s_B \int_{- \infty}^{+ \infty} {\rm d} s_F \ \exp \bigg\{ -
\frac{1}{2} (s^2_B + s^2_F) \bigg\} \nonumber \\
&& \times \exp \bigg\{ - N \ln (E^+ - s_B) + N \ln (E - i s_F)
\bigg\} \nonumber \\
&& \times \bigg\{ \bigg( \frac{1}{E^+ - s_B} + \frac{1}{E - i s_F}
\bigg) \bigg(1 - \frac{N}{(E^+ - s_B)(E - i s_F)} \bigg)
\nonumber \\
&& \qquad - \frac{s_B - i s_F}{(E^+ - s_B)^2(E - i s_F)^2} \bigg\}
\ .
\label{66b}
\ea
We rescale $E \to x = E / \sqrt{N}$ and with it $s_B \to q = s_B /
\sqrt{N}$, $s_F \to p = s_F / \sqrt{N}$, and $j \to j' = j
\sqrt{N}$. Then
\ba
&& \frac{\partial}{\partial j'} \langle {\cal G}(j') \rangle
\bigg|_{j = 0} = \frac{1}{2 \pi} \int_{- \infty}^{+ \infty} {\rm d}
q \int_{- \infty}^{+ \infty} {\rm d} p \ \exp \bigg\{ - \frac{N}{2}
(q^2 + p^2) \bigg\} \nonumber \\
&& \times \exp \bigg\{ - N \ln (x^+ - q) + N \ln (x - i p) \bigg\}
\nonumber \\
&& \times \bigg\{ \bigg( \frac{1}{x^+ - q} + \frac{1}{x - i p} \bigg)
\bigg(1 - \frac{1}{(x^+ - q)(x - i p)} \bigg) \nonumber \\
&& \qquad - \frac{1}{N} \frac{q - i p}{(x^+ - q)^2(x - i p)^2}
\bigg\} \ .
\label{66c}
\ea
That expression agrees with the source terms in Eq.~(\ref{G4}) showing
that the spectral density and its oscillations are in leading order
the same for the UE and for the GUE.

We have reported in Section~\ref{gue} that the supersymmetry approach,
when evaluated exactly, yields the correct finite-$N$ expression for
the spectral density of the GUE in terms of Hermite polynomials. The
approximation used in (\ref{26}) for the UME defines a new
random-matrix ensemble. That ensemble is identical to the GUE except
that the diagonal elements vanish. Therefore, we expect that exact
expressions for the spectral density can be obtained also in this case
from the supersymmetry approach.

Starting from the approximate form of the action given in
Eq.~(\ref{26}), we have performed the integrals as done in the GOE
case.  We have carried the $1 / N$ expansion beyond the leading-order
terms. We have used the result to construct an expansion of the
spectral density of the UME in terms of Hermite polynomials as was
done for the GUE. We have failed to obtain meaningful results. We
ascribe that to the fact that for technical reasons we have not been
able to include all terms in the $1 / N$ expansion. For a truncated
expansion, the spectral density is expected to be singular. We have
not succeeded in separating these singularities from the oscillatory
behaviour of the spectral density.

\section{Graph-theoretical approach}
\label{Graph theory approach}

\subsection{Background and definitions}
\label{gra}

We start with a few definitions of graph-theoretical concepts which
are helpful for the subsequent discussion.

In a \emph{simple} graph with $N$ vertices, any two different vertices
are connected by at most one edge; no edge begins and ends in the same
vertex. In a \emph{complete} graph, every two different vertices are
connected by an edge. The elements $A_{i,j}= A_{j,i}$ of the symmetric
vertex adjacency matrix $A$ of dimension $N$ equal $1$ ($0$) if the
vertices $i$ and $j$ are connected (not connected,
respectively). For a simple graph, $A_{i,i} = 0$. A directed edge $e =
(j, i)$ connects the vertices $j, i$ and has direction $i \rightarrow
j$. The vertex $i$ is the {\it origin} of $e$ and $j$ is its {\it
  terminus}: $i = o(e),\ j = \tau(e)$. The direction of the edge $\hat
e$ is opposite to that of $e$. The number of directed edges equals
$\sum_{i,j}A_{i,j}$. The matrix $\CB$ with elements $B_{e',e} =
\delta_{o(e'),\tau(e)}$ describes the way the vertices are connected
in the space of directed edges. For a complete graph, the matrix $\CB$
has dimension $N (N - 1)$.

A {\it walk} of length $t$ on the graph is defined by a list of
directed edges $e_1, e_2, \cdots, e_t$ where $B_{e_i,e_{i+1}} \ne
0$. For a $t$-\emph{periodic} walk $B_{e_t,e_{1}} \ne 0$. A {\it
  cycle} is the set of all periodic walks that differ only by a cyclic
permutation of their edges. A cycle is primitive if the edge list is
not a repetition of a shorter list.  Writing $J_{e',e}=
\delta_{e',{\hat e}}$ we define the Hashimoto matrix $\CY = \CB - J$
which connects only directed edges that are not reversed to each
other. Thus, while ${\rm Tr} (\CB^t)$ counts the number of
$t$-periodic walks on the graph, ${\rm Tr} (\CY^t)$ counts the number
of $t$-periodic walks where {\it back-tracking} is not allowed.

To use these definitions for the UME we write the phases $\phi_{\mu
  \nu}$ of (\ref{1}) as $\phi_e$ with $e = (\mu, \nu)$. Following
\cite{OrenI,OrenII}, we include these phases in the definitions of the
matrices $\CB$ and $\CY$. Denoting the set of phases of the matrix
$\CM$ by $\Phi$ we define magnetic edge connectivity matrix
$\CB(\Phi)$ and the magnetic Hashimoto matrix $\CY(\Phi)$ as
\begin{eqnarray}
B(\Phi)_{e',e} &=& \delta_{o(e'), \tau(e)} \exp \left\{\frac{i}{2}(\phi_e +
\phi_{e'}) \right\} \ , \nonumber \\
Y(\Phi)_{e',e} &=& (B(\Phi) - J)_{e',e} = Y_{e',e}(0) \exp \left\{
\frac{i}{2}(\phi_e + \phi_{e'}) \right\} \ .
\end{eqnarray}
The term ``magnetic'' relates the phases to a (fictitious) magnetic
field. Contributions to ${\rm Tr} [\CY(\Phi)^n]$ arise from the set
$\Omega_n$ of all $n$-periodic non-backtracking walks on the graph. In
the magnetic case each walk $\omega$ contributes a phase $\Phi_w =
\sum_{e \in w} \phi_e$ so that
\begin{equation}\label{trY}
{\rm Tr} [\CY(\Phi)^n] = \sum_{w \in \Omega_n} \exp \{i\Phi_w\} \ .
\end{equation}

An important identity due to Bass, generalized by Bartholdi and
extended to the magnetic case in Ref.\cite{OrenI}, connects the
spectra of the UME matrix $\CM$ and of $\CY(\Phi)$. The Bass identity
for complete magnetic graphs is valid for any $\eta \in \mathbb{C}$.
With $I^{(k)}$ the identity matrix of dimension $k$, it reads
\begin{equation} \det(\eta I^{(N(N-1))} -\CY(\Phi)) =
(\eta^2-1)^{\frac{N(N-3)}{2}}\det(I^{(N)}(\eta^2+(N-2) ) -\eta \CM)\ .
\label{bass}
\end{equation}
The identity shows that but for a factor $\eta^N
(\eta^2-1)^{\frac{N(N-3)}{2}}$, the characteristic polynomial of
$\CY(\Phi)$ is proportional to the characteristic polynomial of $\CM$
evaluated at $\frac{\eta^2+(N-2)}{\eta}$. The spectra of the two
matrices are, therefore, related. Let $\sigma(\CM) \doteq
\{\lambda_k\}_{k=1}^N$ ($\sigma(\CY(\Phi)) \doteq
\{\eta_r\}_{r=1}^{N(N-1)}$) be the spectrum of $\CM$ (of $\CY(\Phi)$,
respectively). The factor on the right-hand side of (\ref{bass})
vanishes at $\eta = \pm 1$ with multiplicity $N(N-3)/2$ and at the
$2N$ eigenvalues of ($I^{(N)}(\eta^2+(N-2) ) -\eta M$). These can be
expressed in terms of the $\lambda_k$,
\begin{eqnarray}
\eta_k =\left \{
\begin{array}{l}
  \sqrt{N-2} \  \exp\left\{\pm i \arccos \frac{\lambda_k}{ 2\sqrt{(N-2)}} \right\}
\ \ {\rm if} \ |\lambda_k| \le 2\sqrt{N-2} \ , \\
 \sqrt{N-2} \ \exp\left\{\pm \  {\rm arcosh}  \frac{\lambda_k}
{2\sqrt{(N-2)}} \right\} \ \ \ {\rm if} \ |\lambda_k| > 2\sqrt{N-2} \ .
\end{array} \right .
\end{eqnarray}
From the left-hand side of (\ref{bass}) we
see that the nontrivial part of the spectrum of $\CY(\Phi)$ consists
of two sets of points. The first set is confined to the circle of
radius $\sqrt{N-2}$ in the complex plane. It corresponds to the
spectral points of $\CM$ that lie in the interval $[-2\sqrt{N-2},
  2\sqrt{N-2}]$. We write $\sigma^{R}(\CM) = \{\lambda_k:
|\lambda_k|\le 2\sqrt{N-2}\}$. The second set consists of real pairs
$(\eta_+,\eta_-)$ whose product is $(N-2)$. These correspond to the
spectral points in $\sigma^{NR}(\CM) = \sigma(\CM) - \sigma^{R}(\CM)$.
Matrices $\CM$ for which the entire spectrum belongs to
$\sigma^{R}(\CM)$ are referred to as ``Ramanujan'' matrices - a term
which we borrow freely from an analogous situation in the spectra of
$d$-regular graphs. Conversely, matrices for which at least one of the
spectral points does not lie in $\sigma^{R}(\CM)$ are called
``non-Ramanujan''.

For convenience we scale the UME matrices in the following manner;
\begin{equation}\label{Matrix scaling}
\W = \CM/(2\sqrt{N-2}) \ ,
\end{equation}
so that $\epsilon_k = \lambda_k/ (2 \sqrt{N-2})$ are the
eigenvalues\footnote{We note that with this scaling, the
    semi-circle density (the limit of the mean spectral density for
    large $N$) is defined on the interval $[-1,1]$, in contrast to the
    previous sections where different scaling led to the interval
    $[-2,2]$.} Using the connection of the two spectra given by the
  Bass identity~(\ref{bass}), we can write the normalized trace of
  $\CY(\Phi)^n$ as
\begin{eqnarray}
y_n(\Phi) &: =& \frac{1}{N}\frac{ {\rm Tr} [\CY(\Phi)^n]} {(N-2)^{n/2}}
\nonumber \\
&=& \frac{2}{N}\sum_{\epsilon_k \in \sigma^{\rm R}} \cos(n \arccos \epsilon_k )
+ \frac{2}{N}\sum_{\epsilon_k \in \sigma^{\rm NR}} \cosh  (n \ {\rm arcosh}
\epsilon_k ) \nonumber \\
&& \qquad + \ \frac{N(N-3)}{2N} \frac{1+(-1)^n}{(N-2)^{n/2}}
\nonumber\\
&=& \frac{2}{N}\sum_{k=1}^N T_n( \epsilon_k) +\frac{N(N-3)}{2N}
\frac{1+(-1)^n}{(N-2)^{n/2}} \nonumber \\
&=& \frac{2}{N} {\rm Tr} \left[T_n ( \W )\right ] + \frac{ (N-3)}{2}
\frac{1+(-1)^n}{(N-2)^{n/2}} \ .
\label{eq:pre-trace}
\end{eqnarray}
Here $T_n(x)$  are the Chebyshev polynomials of the first kind, given by
\begin{equation}
T_n(x) = \sum_{r=0}^{\floor{\frac{n}{2}}} d^{(n)}_r x^{n - 2r} \ , \
\hspace{20pt} d^{(n)}_r = \frac{n}{2}(-1)^r 2^{n-2r}\frac{(n-r-1)!}{r!(n-2r)!}.
\label{Chebysh}
\end{equation}
Since we sum over eigenvalues in both $\sigma^{\rm R}$ and
$\sigma^{\rm NR}$, equation~(\ref{eq:pre-trace}) is valid for every
matrix $\CM$ in the UME, independently of whether it is Ramanujan or
not. Equation~(\ref{eq:pre-trace}) has been derived in an alternative
manner by Sodin (see e.g. Section 4.2.6 in \cite{Sodin-2016}).

\subsubsection{Expected value of $y_{2n}(\Phi)$}
\label{umeexp}

Using the expressions (\ref{trY}) and (\ref{eq:pre-trace}) we see that
the ensemble average of $y_{2n}(\Phi)$ is given by
\begin{equation}\label{Expected value}
\Av{y_{2n}(\Phi)} = \frac{1}{N}\frac{1}{(N-2)^n}\Av{{\rm Tr} [
    Y(\Phi)^{2n}]} = \frac{1}{N}\frac{1}{(N-2)^n}\sum_{w \in
  \Omega_{2n}} \Av{\exp\{i \Phi_w\}}.
\end{equation}
The expectation value $\Av{\exp\{i \Phi_w\}}$ vanishes unless the total
phase $\Phi_w$ of the non-backtracking walk $w$ obeys $\Phi_w = 0$.
That condition is met only if each edge is traversed both forwards and
backwards the same number of times. The argument implies
$\Av{y_n(\Phi)} = 0$ for all odd $n$.

We display $y_{2 n}$ pictorially in terms of subgraphs as in
Figure~\ref{fig:types123}. Vertices and edges that contribute to $y_{2
  n}$ are depicted as dots and as bonds, respectively. We characterize
the topology of each subgraph in terms of the Betti number $\beta$ (not
to be confused with the matrix index $\beta$ used in
Section~\ref{intro}). Informally, $\beta$ counts the number of
two-dimensional holes in the planar representation of a graph. For
instance, the subgraphs displayed in Figure~\ref{fig:types123} both have
$\beta = 2$. For $n \ll N$ the dominant contributions to
(\ref{Expected value}) come from walks in which each edge is traversed
only once in each direction. Trees ($\beta = 0$) and single loops
($\beta = 1$) cannot occur as neither allows backtracking. Thus, the
dominant contributions come from walks on subgraphs with $\beta=2$. There
are two such subgraphs\footnote{Other forms of subgraphs with two loops do
  exist. These do not contribute to the same order, however, because
  their edges must be traversed more than twice.} called type I and
type II and shown in Figure~\ref{fig:types123}. To describe these we
introduce the following notation. The total number of vertices on a
subgraph is $v$. Vertices located at the intersections are called {\it
  junctions}. The remaining vertices are called {\it simple
  vertices}. The latter are arranged in the form of linear chains that
begin and terminate in a junction. Each such chain is called a {\it
  branch}. In subgraphs of types I and II there are at most three
branches. The number of simple vertices on a branch is denoted by
letters $p, q, r$. The relation between $n$ (the length of the walk),
$v$ and $p, q, r$ depends on the topology of the subgraph (see
Fig.~\ref{fig:types123}).

\vspace{2mm}

\begin{itemize}

\item \textbf{Type I:} \ There is one junction linked to four edges
  and there are two branches carrying $p$ and $q$ vertices,
  respectively. Each branch forms a loop connected to the
  junction. Then $n = v + 1, \ v = p + q + 1$ but $p, q \ge 2$ to
  avoid backscattering. Note that due to the
    constraints we must have $n \geq 6$.

\item \textbf{Type II:} \ There are two junctions linked to three
  edges each and three branches containing $p, q, r$ vertices,
  respectively. Then $n = v + 1, \ v = p + q + r + 2$. We may have
  either $p, q, r \geq 1$, or $p, q \geq 1$ and $r = 0$ (cyclic). Note
  that even though a periodic walk of length $2n$ traverses the entire
  subgraph with each edge traversed in both directions, the walk of
  length $n$ is not a periodic walk, in contrast to walks in type
  I. Note that due to the constraints we must have $n
    \geq 6$ if $p, q, r \geq 1$ or $n \geq 5$ if one of $p,q$ or $r$
    is zero.

\end{itemize}

\begin{figure}[ht]
\centerline{\includegraphics[width=0.9\textwidth]{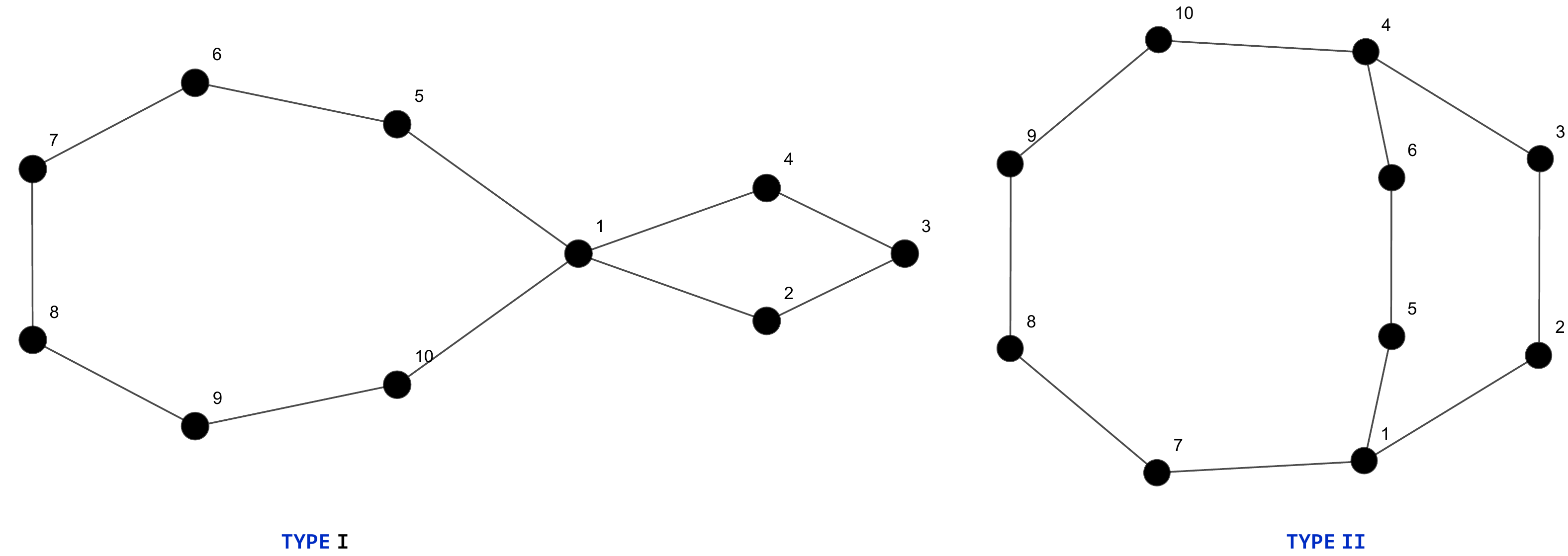}}
\caption{Examples of the two types of subgraphs. Type I with $a=1$,
  $p^+ = (2,3,4)$ and $q^+ = (5,6,7,8,9,10)$ and type II subgraph with
  $a=1$, $b=4$, $q^+=(2,3)$, $p^+ = (5,6)$ and $r^+ = (7,8,9,10)$.
  For explanations, see text}
\label{fig:types123}
\end{figure}

We calculate the contributions to (\ref{Expected value}) from type I
and type II subgraphs. For type I the contribution is determined by the
total number $\Av{{\rm Tr} [\CY(\Phi)^{2n}]}_{\rm I}$ of walks $w =
(e_1,\ldots,e_{2t})$ that trace out the subgraph such that every edge is
traversed once in each direction. Let $a$ denote the vertex at the
junction, $p^+ = (\mu_1,\ldots,\mu_{p})$ the ordered vertices on the
$p$-branch, $p^- = (\mu_{p},\ldots,\mu_1)$ the reversed order, and
analogously for $q$. With $e_1 = (a,\mu_1)$ fixed, there are two
possible traversals of the subgraph given by $w = ap^+aq^+ap^-aq^-$ and
$w= ap^+aq^-ap^-aq^+$. There are $2n$ possible directed edges from
where to start, giving a total of $4n$ walks for each labelled
subgraph. With the total number of vertices given by the dimension $N$ of
$\CM$, there are $N!/(N-v)!$ ways of choosing the vertices. We have to
sum over all values of $(p, q)$ subject to the constraints $p + q = v
- 1$ and $p, q \geq 2$. To avoid overcounting (due to reversing of the
orientations of walks in each branch and the exchange of the two
branches) we have to divide by $s_{\rm I} = 2\times 2 \times 2 =
8$. That gives
\begin{eqnarray}\label{Type I contr}
\Av{{\rm Tr} [\CY(\Phi)^{2n}]}_{\rm I} &=&
\frac{N!}{(N-v)!}\frac{4n}{s_{\rm I}}| \{(p,q) : p,q \geq 2, p+q = v-1\}|
\nonumber \\
&=& \frac{N!}{(N-v)!}\frac{4n}{s_{\rm I}}(v-4),
\end{eqnarray}
which is valid for $n \geq 6$.

To count the number $\Av{{\rm Tr} [\CY(\Phi)^{2n}]}_{\rm II}$ of walks
$w$ that traverse subgraphs of type II, we denote the two vertices at the
junctions by $a$ and $b$ and the branches by $p^{\pm},q^{\pm},r^{\pm}$
similarly as before. For every labelled subgraph, there are two possible
traversals $ap^+bq^-ar^+bp^-aq^+br^-$ and $ap^+b r^-aq^+b p^-ar^+b
q^-$ and $2n$ directed edges from where to start, giving a total of
$4n$ traversals. There are $N!/(N-v)!$ ways of choosing the vertices.
In addition, we have to sum over $p,q,r$ subject to the constraints
(either $p, q, r \geq 1$, or $p, q \geq 1$ and $r = 0$ (cyclic)). To
avoid overcounting we must divide by $s_{\rm II} = 3! \times 2 =
12$. This comes from 3! ways of exchanging the three branches and
from the reflection of the subgraph about its centre
(i.e. exchanging $a \leftrightarrow b$ and relabelling $(p^+,q^+,r^+)
\leftrightarrow (p^-,q^-,r^-)$). Thus, for $n \geq 6$ and $v = n-1$,
\begin{eqnarray}\label{Type II a contr}
\Av{{\rm Tr} [\CY(\Phi)^{2n}]}_{\rm II} & = & \frac{N!}{(N-v)!}
\frac{4n}{s_{\rm II}}\bigg( | \{(p,q,r) : p,q,r
\geq 1, p+q+r = v-2\}| \nonumber \\ & & \hspace{70pt} + 3| \{(p,q) :
p,q \geq 1, p+q = v-2\}|\bigg) \nonumber \\ & = &
\frac{N!}{(N-v)!}\frac{4n}{s_{\rm II}}\bigg(\frac{(v-4)(v-3)}{2} +
3(v-3) \bigg) \ .
\end{eqnarray}
For $n=5$ we only keep the second term on the
    right-hand side and so
\begin{equation}\label{Type II b contr}
\Av{{\rm Tr} [\CY(\Phi)^{2n}]}_{\rm II} = \frac{N!}{(N-v)!}\frac{4n}{s_{\rm II}}
3(v-3), \qquad n =5
\end{equation}
For $n=1,\ldots,4$ we have $\Av{{\rm Tr} [\CY(\Phi)^{2n}]} = 0$ since
there do not exist any non-backtracking paths of length 8 or less in
which all the edges are traversed the same number of times in both
directions. We may now  combine the expressions (\ref{Type I
  contr}), (\ref{Type II a contr}) and (\ref{Type II b contr}) for $n \geq 5$. So, using that $v = n-1$ and $N!/(N-v)! = N^{n-1} - \O(N^{n-2})$
for $N \gg v$ we get
\begin{equation}
\fl \Av{{\rm Tr} [\CY(\Phi)^{2n}]} = \left\{\begin{array}{ll}
N^{n-1} n(n-4)  + \O(N^{n-2}) & n = 5 \ , \\
\frac{nN^{n-1}}{6} \bigg[(n+1)(n-4) + 3(n-5)\bigg] + \O(N^{n-2}) & n \geq 6
\ .
\end{array}\right. 
\end{equation}
The relation (\ref{Expected value}) between $\Av{{\rm Tr}
  [\CY(\Phi)^{2n}]}$ and $\Av{y_{2n}(\Phi)}$ then leads to
\begin{equation}\label{y expectation}
\Av{y_{2n}(\Phi)} = \left\{\begin{array}{ll}
\frac{n}{N^2}(n-4)  + \O(N^{-3}) & n=5 \ , \\
\frac{n}{6N^2} \bigg[ (n+1)(n-4) + 3(n-5)\bigg] + \O(N^{-3}) & n \geq 6 \ .
\end{array}\right.
\end{equation}
Therefore, from Eq.~(\ref{eq:pre-trace}) the expectation
value of the trace of the Chebyshev polynomial is
\begin{equation}\label{Chebyshev subleading}
\Av{{\rm Tr} [T_{2n}(\W)]} =  \left\{\begin{array}{ll}
 \frac{n}{2N}(n-4) + \O(N^{-2}) & n = 5 \ , \\
\frac{n}{12N} \bigg[ (n+1)(n-1) -18\bigg] + \O(N^{-2}) & n \geq 5 \ ,
\end{array}\right.
\end{equation}
whereas for $n<5$ the relation (\ref{eq:pre-trace}) gives $\Av{{\rm
    Tr} [T_{2n}(\W)]} = -
\frac{N}{2}\frac{(N-3)}{(N-2)^n}$.

In comparison to (\ref{Chebyshev subleading}) the equivalent
expectation for the GUE can be obtained by combining the result
(\ref{GUE moments}) with the form of the Chebyshev polynomial
(\ref{Chebysh}) to obtain
\begin{equation}\label{Chebyshev GUE subleading}
\Av{\Tr\left[T_{2n}\left(\frac{\H}{2\sqrt{N}}\right)\right]}_{\rm GUE}
= - \frac{N}{2}\delta_{n,1} + \frac{n(n^2-1)}{12N} +
\O(N^{-3})
\end{equation}
We highlight that for the GUE there is no constant term in
  the $1/N$ expansion. One may see immediately why this is the case
  from the form of the moments (\ref{GUE moments}), which after
  multiplying by a factor of $N$ do not contain any constant
  term\footnote{This is not the case for other Gaussian
    $\beta$-ensembles however.}. Whereas for $n=2$ the UME satisfies
  $\Av{\Tr[T_{2n}(\W)]} \to -\frac{1}{2}$ as $N \to \infty$. We could
  remove this constant term for $n=2$ by changing the scaling of $\W$
  from $(N-2)^{-1/2}$, which arises naturally in
  (\ref{eq:pre-trace}). However, choosing another scaling of the form
  $(N-c)^{-1/2}$ for some constant $c$ will induce an order $\O(1/N)$
  correction (from the leading term, which is 0 for all $n>1$) which
  will result in $\Av{\Tr[T_{2n}(\W)]}$ having constant terms for
  other values of $n$.
  
\subsubsection{Correlations of $y_n(\Phi)$.}\label{Correlations of yn}
From expression~(\ref{trY}) the covariance of the traces of powers of
$\CY(\Phi)$ is given by
\begin{eqnarray}\label{Correlations}
\fl \Cov({\rm Tr} [ \CY(\Phi)^n], {\rm Tr} [ \CY(\Phi)^m]) &:= & \Av{{\rm Tr} [\CY(\Phi)^n] {\rm Tr} [\CY(\Phi)^m]} -
\Av{{\rm Tr}[ \CY(\Phi)^n]} \Av{{\rm Tr} [\CY(\Phi)^m]} \nonumber \\
\fl &  = &  \sum_{w \in \Omega_n} \sum_{w' \in \Omega_m} \Av{\exp\{i(\Phi_w -
\Phi_{w'})\}} - \Av{\exp\{i\Phi_w\}} \Av{\exp\{-i\Phi_{w'}\}} \nonumber \\
\fl & = & |\Omega_{n,m}(0)|.
\end{eqnarray}
Here $\Omega_{n,m}(0) := \{(w,w') \in (\Omega_n,\Omega_m) : \Phi_w =
-\Phi_{w'} \neq 0 \ \forall \ \Phi\}$ denotes the set of pairs of non
back-tracking walks that have nonvanishing opposite phases for every
$\CM \in$ UME.

\vspace{2mm}

\noindent {\it Variance of $y_n(\Phi)$}

\vspace{2mm}

\noindent The dominant contribution to $|\Omega_{n,n}(0)|$ comes from
pairs of walks $w,w'$ that reside on a subgraph with a single loop in
which every edge is traversed exactly once by $w$ and in the opposite
direction by $w'$. Hence there are $v = n$ vertices on the
subgraph. The path $w$ can start from each edge and in either
direction, giving $2n$ possible starting positions and then there are
$N!/(N-v)!$ ways of labelling the vertices. That, however, overcounts
by a factor $s_{\beta=1} = 2n$ since all starting position of $w$ can
also be obtained by relabelling the vertices. For every $w$ we have
$n$ possible walks $w'$ which gives, for $n \geq 3$

\begin{equation}\label{beta 1 contribution}
\fl \Var({\rm Tr} [\CY(\Phi)^n])_{\beta = 1} =
\frac{N!}{(N-v)!}\frac{2n^2}{s_{\beta=1}} 
= n\bigg(N^n - N^{n-1}\frac{n(n-1)}{2} + \O(N^{n-2})\bigg).
\end{equation}

The next-to-leading-order contribution to $|\Omega_{n,n}(0)|$ comes
from pairs of non-backtracking walks on subgraphs containing two loops
($\beta = 2$). As shown in the previous section there exist two types
of subgraphs, for which we require $w$ to traverse each edge of these
subgraphs precisely once and $w'$ to traverse the same subgraph in the
opposite direction. That is not possible for type II subgraphs as
these do not support non-backtracking walks in which the loops are
traversed only once. Thus leaving subgraphs of type I.

To obtain the contribution from the type I subgraphs (where $v =
n-1$), we see that if a walk $w$ starts at a particular edge then
there are two possible traversals of the subgraph, which in the
notation of the previous section, are given by $w= ap^+aq^+$ and
$w=ap^+aq^-$. There are $n$ possible starting edges for the walk $w$
and hence a further $n$ possible choices for $w'$. We may also relabel
the vertices, which gives a factor of $N!/(N-v)!$ but must gain
mitigate for the overcounting. Thus we must divide through by a factor
of 4, coming from the possibility of swapping the two loops and the
two different possible traversals of the subgraph by $w$. So
altogether, for $n \geq 6$
\begin{eqnarray}\label{beta 2 contribution}
\Var({\rm Tr} [\CY^n])_{\beta = 2} & = &
\frac{N!}{(N-v)!}\frac{2n^2}{4}| \{(p,q) : p,q \geq 2, p+q =
v-1\}| \nonumber \\ & = & \frac{n^2}{2} (n-5)(N^{n-1} + \O(N^{n-2})) \ .
\end{eqnarray}
Combining (\ref{beta 1 contribution}) and (\ref{beta 2 contribution})
leads to
\[
\Var({\rm Tr} [\CY(\Phi)^n]) = nN^n - 2n^2 N^{n-1} +
\O(N^{n-2}) \ .
\]
Therefore, using that $(N - 2)^{-n} = N^{-n} - 2N^{-n-1} + \ldots$ we
have from (\ref{eq:pre-trace}) that
\begin{eqnarray}
\Var(y_n(\Phi)) & = & \frac{1}{N^2}\Var({\rm Tr} [\CY(\Phi)^n])(N^{-n}
- 2N^{-n-1} + \O(N^{-n-2})) \nonumber \\ & = & \frac{n}{N^2} -
\frac{2n(n+1)}{N^3} + \O(N^{-4}).
\end{eqnarray}

This in turn implies that
\[
\Var\left({\rm Tr} [T_n(\W)]\right) = \frac{n}{4} - \frac{n(n+1)}{2
  N} + \O(N^{-2}).
\]
The leading term in this expression coincides with the result of
Johansson \cite{Johansson-1998} for the GUE. The manner in which it
has been derived (i.e., counting the number of non-backtracking walks
on single loops) is the same as in \cite{Sodin-2016} (see also
\cite{Anderson-2009,Kusalik-2007,Schenker-2007}). In fact,
  this approach is capable of showing that all the joint moments of
  the type $\Av{ {\rm Tr} [T_{1}(\W)]^{a_1} {\rm
      Tr}[T_{2}(\W)]^{a_2}\ldots {\rm Tr}[T_{k}(\W)]^{a_k}}$ with $a_i
  >0$, coincide in the large $N$ limit with averages of the form
  $\Av{Z_1^{a_1} Z_2^{a_2} \ldots Z_k^{a_k}} = \Av{Z_1^{a_1}}\Av{
    Z_2^{a_2}} \ldots \Av{Z_k^{a_k}}$, where the $Z_n$ are independent
  and identically distributed Gaussian random variables with zero mean
  and variance $\sigma_n^2 = n / 4$.
\vspace{2mm}

\noindent {\it Covariance of $y_n(\Phi)$ and $y_m(\Phi)$}

\vspace{2mm}

\noindent The covariance is obtained by setting $n \neq m$ (we take $n
> m$ without loss of generality) in (\ref{Correlations}) and computing
the number of pairs of non-backtracking walks $|\Omega_{n,m}(0)|$ of
length $n$ and $m$ which retrace each other. Obviously $\Cov({\rm Tr}
[\CY(\Phi)^n],{\rm Tr} [\CY(\Phi)^m]) =0$ if $(n-m)$ is odd since then
the phases $\Phi_w$ and $\Phi_{w'}$ cannot be equal in general.

For even $(n-m)$ the leading contribution comes from subgraphs of type
I in which $w \in \Omega_n$ traverses one of the loops twice in
opposite directions whilst $w' \in \Omega_m$ only traverses the other
loop once. For example $w = ap^+aq^+ap^-$ and $w' = aq^-$. Thus one
loop contains $m$ edges, the other $(n-m)/2$ edges and there are a
total of $v = (n+m-2)/2$ vertices. The number of such pairs $(w,w')$
is given by noting that the $w$ walk has $2n$ possible ways of
traversing the subgraph, given by the 2 possible orders of traversing
the $(n-m)/2$ loop and the $n$ possible starting edges. Then for each
$w$ there are $m$ ways of choosing the $w'$ walk. Relabelling the
vertices also gives a factor of $N!/(N-v)!$ however we must then
divide by a factor of $2$ to account for reversing the orientation of
the loop that is traversed twice by $w$. Altogether this gives
\begin{eqnarray}
\fl \Cov({\rm Tr} [\CY(\Phi)^n], {\rm Tr} [\CY(\Phi)^m]) & = & N^v \frac{2nm}{2}+ \O(N^{v-1}), \nonumber \\
\fl & = &  nmN^{(n+m-2)/2} + \O(N^{(n+m)/2}) \ \ |n-m| \geq 6, {\rm even}
\end{eqnarray}
which in turn leads to
\[
\Cov(y_n(\Phi),y_m(\Phi)) = \frac{nm}{N^3} + \O(N^{- 4}), \qquad |n-m| \geq 6, {\rm even}
\]

\subsection {Spectral density and two-point correlation function}

\subsubsection{Trace formula for the spectral density}
\label{Trace formula}

Given any matrix $\CM \in {\rm UME}$, the density of the scaled
eigenvalues $\epsilon_k$ (see (\ref{Matrix scaling}))
\begin{equation}
\rho_{\CM}(\epsilon)=\frac{1}{N}\sum_{k=1}^N\delta(\epsilon-\epsilon_k)
\end{equation}
is a distribution which we study by applying it to test functions
(observables) that are analytic on the entire real line. We restrict
our attention to this space of functions because the maximum (scaled)
spectral radius (achieved by setting $M_{\mu\mu} = 1$ for all $\mu
\neq \nu$) for the UME is given by $[(N-1) / (2\sqrt{N-2})] \thicksim
\sqrt{N}$. 

Let $f(x)$ be an allowed test function. It can be expanded in terms of
Chebyshev polynomials:
\begin{equation}
f(\epsilon)=\sum_{m=0}^{\infty} f_m T_m(\epsilon) \ \ \ ;
\ \ f_m=\frac{2-\delta_{m,0}}{\pi}\int_{-1}^{1}\frac{f(\epsilon)
T_m(\epsilon)}{\sqrt{1-\epsilon^2}}{\rm d}\epsilon\ .
\label{def:fm}
\end{equation}
The series converges on the entire real line, since it is a
rearrangement of the Taylor expansion. We note, however, that the
coefficients $f_m$ are derived from the restriction of $f(\epsilon)$
to the interval $[-1,1]$.

Recalling (\ref{eq:pre-trace}) and the fact that $y_0(\Phi)=(N-1),
y_1(\Phi)=y_2(\Phi)=0$ we can write
\begin{eqnarray}
\frac{1}{N} {\rm Tr} \left[ f\left(\W \right)\right] & = &
\frac{1}{N}\sum_{n=0}^{\infty}f_n {\rm Tr} [T_n(\W)] \nonumber \\
& = & \frac{1}{2}\sum_{n=3}^{\infty}y_n(\Phi) f_n + \frac{(N-1)}{2}f_0
-\frac{(N-3)}{4}\sum_{n=0}^{\infty} f_n\frac{1+(-1)^n}{(N-2)^{\frac{n}{2}}} \ .
\nonumber
\end{eqnarray}
The sum over $y_n(\Phi)f_n$ converges because $|y_n|<(2 \sqrt{N-1}
)^n$. Therefore, inserting the expression (\ref{def:fm}) for the
coefficients $f_n$, using the absolute convergence of the series to
exchange summation and integration, and noting that $T_0(\epsilon) =
1$ leads to
\begin{eqnarray}\label{spectraltest}
\frac{1}{N} {\rm Tr} [f(\W)] & = &
\frac{1}{2}\sum_{n=3}^{\infty}y_n(\Phi) f_n \nonumber \\
\qquad &+& \frac{1}{\pi}\int_{-1}^{1}d\epsilon
\frac{f(\epsilon)}{\sqrt{1-\epsilon^2}} \Bigg\{ (N-2)
-\frac{(N-3)}{2}\sum_{n=0}^{\infty} T_n(\epsilon)
\frac{1+(-1)^n}{(N-2)^{\frac{n}{2}}} \Bigg\} \nonumber \\ 
& = & \frac{1}{2}\sum_{n=3}^{\infty} y_n(\Phi) f_n+ \int_{-1}^{1} d\epsilon
\left \{\frac{2}{\pi} \sqrt{1 - \epsilon^2} \frac{1}{1 + \frac{1}{N -
    2} - \frac{4}{N - 1} \epsilon^2}\right \} f(\epsilon) \ ,
\nonumber \\
\end{eqnarray}
where going to the final line we have made use of the identity
\[
\sum_{n=0}^{\infty} T_n(x)(y^n +(-y)^n) =
\frac{2(1+y^2)-4x^2y^2}{(1+y^2)^2-4x^2y^2} \ , \ \ \ \ \ |x|,|y|\le 1.
\]
For appropriate test functions $f(x)$ and for any $\CM \in$ UME
(\ref{spectraltest}) is an exact, absolutely convergent trace formula
and provides the correct manner in which one can apply the formal
trace formulae derived in \cite{OrenI}.

The absolute convergence of (\ref {spectraltest}) permits computing
the ensemble average term by term. Due to (\ref{y expectation}), the
infinite sum over $\Av{y_n(\Phi)}f_n$ is then of lower order in $N$
than the leading term, the integral, which does not depend on any
periodic walk information. Therefore, in the limit of large $N$, the
expression within the curly brackets in the integrand can be
interpreted as the mean spectral density and indeed, in this limit, it
converges to the semi-circle density.

\subsubsection {Mean spectral density for finite $N$}

The term appearing in the curly brackets in (\ref{spectraltest}),
\begin{equation}\label{Mean density}
\langle\rho(\epsilon)\rangle=\frac{2}{\pi}
\sqrt{1 - \epsilon^2} \
\frac{1}{1 + \frac{1}{N - 2} - \frac{4}{N - 1} \epsilon^2},
\end{equation}
can be identified as the mean spectral density. It includes $1 / N$
corrections to the semi-circle law. However, the integration domain is
the interval $[-1,1]$ so that the possible contributions for finite
$N$ due to the non-Ramanujan part of the spectrum can only come from
the periodic orbit sum in (\ref {spectraltest}). In Figure \ref{fig2}
we compare the numerical mean density for $N=10$ with the expression
given in (\ref{Mean density}), as well as the semi-circle
density. None accounts for the oscillations that are due to the
contributions from the periodic orbit sum.

To proceed further we introduce a family of $\delta$-like functions:
\begin{equation}
\delta_{N^{\star}}(x;\xi) = \frac{1}{1+\frac{\pi}{2}
  \sqrt{1-\xi^2}}\sum_{m=0}^{N^{\star}} T_m(x)T_m(\xi) \ ,
\ \ \ \ |\xi| < 1-\frac{c}{{N^{\star}}},
\label{deltam}
\end{equation}
where $c$ is a numerical constant, and ${N^{\star}}$ a large but
finite integer. We consider $\delta_{N^{\star}}(x;\xi)$ a function of
$x$ that depends on the parameter $\xi$, with $\xi$ restricted so as
to have a finite distance from the end points of the interval $[-1,
  1]$. The function $\delta_{N^{\star}}(x;\xi)$ is concentrated about
the point $x=\xi$ where it takes its maximum value $(N^{\star} + 1) /
(2(1+\frac{\pi}{2} \sqrt{1-\xi^2}))$ with full width at half maximum
$[2\sqrt{3} / N^{\star}] \sqrt{1-\xi^2}$.  The integral over the
domain $[-1,1]$ is unity up to a correction of order
$\mathcal{O}(1/N^{\star})$. For $x$ values sufficiently far from
$\xi$, $\delta_{N^{\star}}(x;\xi)$ oscillates about $0$ with mean
amplitude of order $1$. The function $\delta_{N^{\star}}(x;\xi)$ is a
polynomial in $x$ and therefore it belongs to the class of test
functions relevant to our discussion. Because of these properties the
function
\begin{equation}
\bar{\rho}_{N^{\star}}(\xi) = \frac{1}{N} \Av{{\rm Tr}
  \left[\delta_{N^{\star}}\left(\W;\xi\right)\right]} \ ,
\ \ \ \ |\xi|<1-\frac{c}{{N^{\star}}}
\end{equation}
provides a smooth mean spectral density. The smoothing is done over
spectral intervals of order $1 / N^{\star}$ that lie within the domain
$|\xi|<1-\frac{c}{{N^{\star}}}$. The coefficients $f_m$ for
$\bar{\rho}_{N^{\star}}(\xi)$ are obviously proportional to $T_m(\xi)$
for $m \le {N^{\star}}$ and zero otherwise. Thus the oscillatory part
of the mean spectral density is
\begin{equation}
\hspace{-10mm} \bar{\rho}_{N^{\star}}(\xi) - \langle \rho (\xi)
\rangle = \frac{1}{ \frac{2}{\pi} + \sqrt{1-\xi^2}}\sum_{m=3}^{
  \lfloor {N^{\star}}/2 \rfloor }\Av{y_{2m}(\Phi)}T_{2m}(\xi)\left(1
+\O\left(\frac{1}{N^{\star}}\right)\right) \ .
\label{smoothrow}
\end{equation}
Our numerical results (Figure~\ref{fig2}) suggest that in the interval
$[-1,1]$, the oscillatory part of the mean spectral density possesses
$N$ oscillations about the mean, and this indicates that the parameter
$N^{\star}$ should be at least $N$. Thus, in order to use
(\ref{smoothrow}) to match the data, one needs to compute
$\Av{y_{2m}(\Phi)}$ at least for all $m \le N$. Unfortunately the
combinatorial computation (\ref{y expectation}), which provides an
estimate for $\Av{y_{2m}(\Phi)}$, is valid only for $m\le \sqrt{N}$.
\begin{figure}[ht]
\centerline{\includegraphics[width=0.7\textwidth]{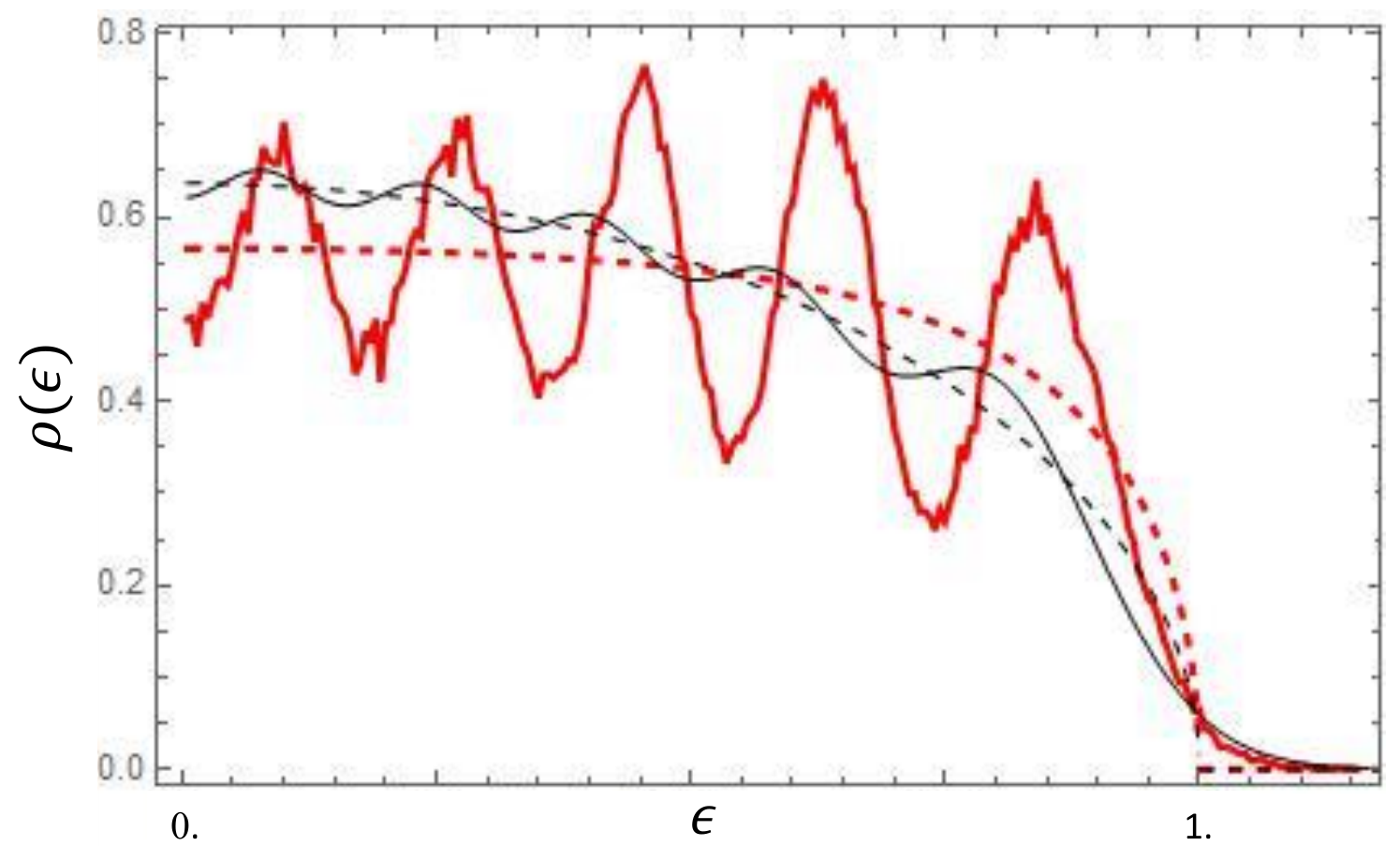}}
\caption{spectral densities for N=10. Thick full line: numerical data
  (20,000 realizations); Thick dashed line: the modified mean spectral
  density (\ref{Mean density}); Thin full line: GUE density; Thin
  dashed line: semi-circle density. Here $\epsilon$ is the normalized
  spectral parameter so that the asymptotic support of the densities
  is $[-1,1]$.}
\label{fig2}
\end{figure}

\subsubsection {Spectral form factor}

For the discussion in the present Section it is convenient to map the
spectral interval $-1\le \epsilon \le 1$ onto the unit circle by
$\theta = 2\arccos(\epsilon)$. We restrict the attention to Ramanujan
matrices. Evaluating, formally, $N^{-1}{\rm Tr} [\delta(\epsilon -
  \W)]$ in (\ref{spectraltest}) gives us the spectral density, which
may be written as
\[
\rho_{\CM}(\theta) =  \Av{ \rho_{\CM}(\theta)} +\tilde{\rho}_{\CM}(\theta).
\]
The mean term here is given by the transformation of (\ref{Mean
  density}) and the oscillatory term by the sum over the $y_n(\Phi)$
in (\ref{spectraltest})
\begin{eqnarray}
\label{rhophi}
\left \langle \rho_{\CM}(\theta) \right \rangle &=& \frac{1}{\pi}
\sin^2\left(\frac{\theta}{2}\right) \frac{1}{1 + \frac{1}{N - 2} -
  \frac{4}{N - 1} \cos^2(\frac{\theta}{2})} \ , \nonumber \\ \tilde
\rho_{\CM}(\theta) &=& \frac{1}{2\pi} \sum_{n=3}^{\infty} y_t(\Phi)
\cos\left( n \frac{\theta}{2} \right) \ .
\end{eqnarray}
The spectral two-point correlation function is defined as
\[
R_2(\eta) = \frac{1}{2 \pi N^2} \sum_{i \neq j} \left \langle
\delta(\eta - (\theta_j - \theta_i) \right \rangle \ .
\]
Using Eq.~(\ref{rhophi}) we can write this as
\begin{eqnarray}
 R_2(\eta) &=& \int_{-\pi}^{\pi} \frac{{\rm d} \theta}{2 \pi} \left
 \langle \tilde \rho_{\CM}\left(\theta + \frac{\eta}{2}\right) \tilde
 \rho_{\CM}\left(\phi - \frac{\eta}{2}\right) \right \rangle -
 \frac{1}{2 \pi N} \delta(\eta)\nonumber \\ &+& \int_{-\pi}^{\pi}
 \frac{{\rm d} \theta}{2\pi} \rho_{\CM} \left(\theta +
 \frac{\eta}{2}\right) \rho_{\CM}\left(\theta - \frac{\eta}{2}\right),
\end{eqnarray}
where we have substituted $\rho(\phi)$ for $\left \langle
\rho_{\CM}(\theta)\right \rangle$.  Now,
\[
\int_0^{2\pi} \frac{{\rm d} \theta}{2 \pi} \left \langle
\tilde{\rho}_{\CM} \left(\theta + \frac{\eta}{2}\right)
\tilde{\rho}_{\CM} \left(\theta - \frac{\eta}{2}\right) \right \rangle
\]
\[
= \frac{1}{4\pi^2} \sum_{n,m} \Av{y_n(\Phi) y_m(\Phi)} \int_0^{2\pi}
\frac{{\rm d} \theta}{2 \pi} \cos\left(\frac{n}{2}\left( \phi +
\frac{\eta}{2} \right)\right)\cos\left(\frac{m}{2} \left(\phi -
\frac{\eta}{2} \right) \right)
\]
\begin{equation}
= \frac{1}{8 \pi^2} \sum_n \Av{y_n(\Phi)^2}
\cos\left(\frac{n\eta}{2}\right)\ .
\end{equation}
To leading order,
\[
\int_{-\pi}^{\pi} \frac{{\rm d} \theta}{2 \pi}
\rho_{\CM}\left(\theta + \frac{\eta}{2}\right) \rho_{\CM} \left(\theta
- \frac{\eta}{2}\right) = \frac{1}{4 \pi^2} + \frac{\cos(\eta)}{8
  \pi^2} \bigg(1 - \frac{2}{N - 2} \bigg) + \O \left(\frac{1}{N^2}
\right) \ .
\]
The spectral form factor is defined as
\begin{equation}
\label{formf}
 K_2(t;N) = \Av{\frac{1}{N} \left | \sum_j \exp\{i\theta_j t\} \right |^2}
 - N \ .
 \end{equation}
Substituting and collecting the terms, we get
\begin{equation}
\label{formf1}
K_2(t; N) = \frac{N}{4} \bigg(1 - \frac{2}{N - 2} \bigg) \delta_{n, 1}
+ \frac{N}{4} \left \langle y_{2t}(\Phi)^2 \right \rangle \ .
\end{equation}
The spectral information used to calculate the form
factor~(\ref{formf1}) consists of the entire set of values within the
asymptotic support. The spectral density is not constant. Therefore,
Eq.~(\ref{formf1}) cannot be compared directly with the GUE ``local''
form factor. The latter is defined for the unfolded spectrum with
constant mean spectral density. For a comparison one has to transform
the GUE result using a convolution integral introduced previously in
\cite{OrenII}. The resulting GUE expression for $K_2(t,N)$ is a
function of the scaled variable $\tau = t / N$ and is given by
\begin{equation}
\fl
K_2 ^{\rm GUE}(\tau) = \left \{\begin {array}{ll}
 \tau (1 - \frac{2}{\pi} \arcsin \sqrt{\frac{\tau}{2}})
 + \frac{1}{\pi} \left(2 \arcsin \sqrt{\frac{\tau}{2}} - \sin (2\arcsin
\sqrt{\frac{\tau}{2}}) \right) \ , \ \  & \tau \le 2 \ , \\
 1 \ ,   & \tau > 2 \ .
\end{array} \right .
\end{equation}
We observe that the leading term in (\ref{formf1}) is $\tau$-inherited
from the GUE expression. The next-order terms consist of odd powers of
$\tau^{\frac{1}{2}}$. The numerical data and the expression for
$K_2^{\rm GUE}(t / N)$ above are compared in Fig. \ref{fig:formfactor}
for $N = 20$. The agreement is very good.

\begin{figure}[ht]
\centerline{\includegraphics[width=0.6\textwidth]{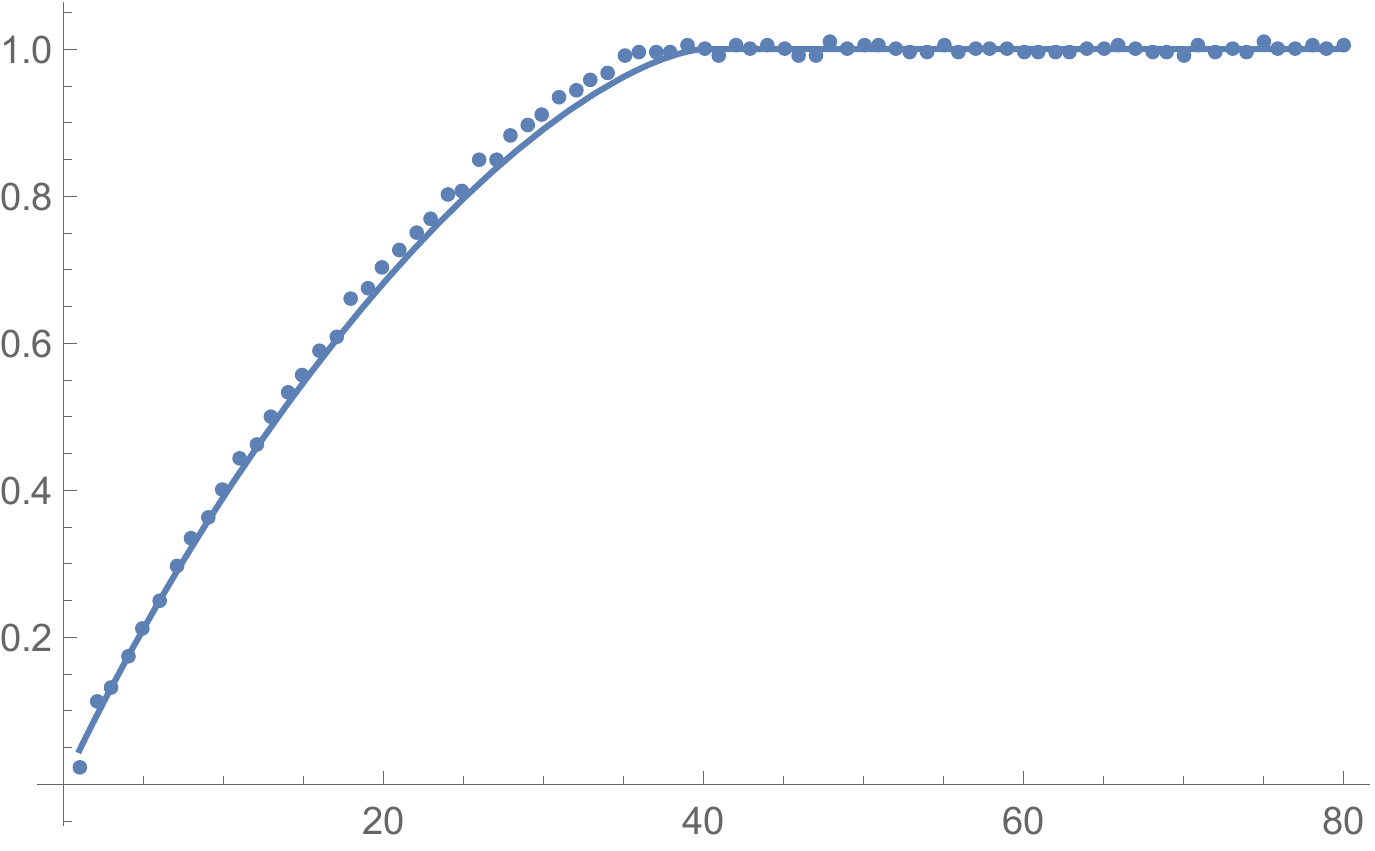}}
\caption{The numerical form factor (dots) and the GUE expression
  (line) for $N=20$ }
\label{fig:formfactor}
\end{figure}

\section{Brownian motion approach}\label{BM approach}

In the previous section we obtained expressions for the average
spectral moments of the UME in a basis of Chebyshev polynomials (see
Eqn. (\ref{Chebyshev subleading})). By inverting the relation
(\ref{Chebysh}) one may obtain the monomials from this basis and hence
obtain expressions for the average value of the expressions
$\frac{1}{N}\Av{\Tr[\W^{2n}]}$, whose leading terms correspond to the
moments of the semicircle distribution $\frac{2}{\pi}\sqrt{1-
  \lambda^2}$. The result (\ref{Chebyshev subleading}) therefore
allows us to obtain the deterministic deviations for averages
$\frac{1}{N}\Av{\Tr[f(\W)]}$ of polynomial test functions over the UME
from the average over the semicircle distribution, i.e.
\[
\frac{2}{N\pi}\int_{-1}^{1} f(\lambda) \sqrt{1- \lambda^2} d \lambda.
\] 
In the present section we turn to analysing the behaviour of
fluctuations of the traces of Chebyshev polynomials about their
mean. In particular we will show, using a Brownian motion approach,
that in the large $N$ limit the random variables $\Tr[T_n(\W)]$ behave
like independent Gaussian random variables (as was briefly highlighted
at the end of Section \ref{Correlations of yn}) and provide bounds on
the rates of this convergence utilising results derived from Stein's
method \cite{Meckes-2009}. This is
  encapsulated in our main result of this section, which is Theorem
  \ref{Wasserstein distance estimation}.

\vspace{5pt}

Let us therefore first define our centred traces of Chebyshev
polynomials by removing their mean
\begin{eqnarray}\label{Centered Chebyshev}
\fl F_n(\CM) & :=&  \Tr[T_n(\W)] - \Av{\Tr[T_n(\W)]} = \frac{1}{2(N-2)^{n/2}} \left(\Tr[\CY(\Phi)^n] - \Av{\Tr[\CY(\Phi)^n]}\right) \\
\fl & =& \frac{1}{2(N-2)^{n/2}}\sum_{w \in \Omega_n} (\exp\{i \Phi_w\} - \Av{\exp\{i \Phi_w\} }) = \frac{1}{2(N-2)^{n/2}}\sum_{w \in \Lambda_n} \exp\{i \Phi_w\}. \nonumber 
\end{eqnarray}
Here $\Lambda_n := \Omega_n \setminus \{ w \in \Omega_n : \Phi_w = 0
\ \forall \ \Phi \in {\rm UME} \}$ is the set of periodic
non-backtracking walks of length $n$ whose phases are not identically
0 for every member of the UME.

As was mentioned in the introduction, there are many examples of
random matrix ensembles in which the linear-statistic $L_f(\H)$
converges in distribution to a Gaussian random variable with a
universal variance (which we do not state here but may be found in
\cite{Johansson-1998,Schenker-2007} for example) as $N$ becomes
large. Many works have sought to establish regularity conditions for
these linear statistics, i.e. finding the class of test functions for
which $L_f(\H)$ converges to a Gaussian random variable in the large
$N$ limit (see e.g. \cite{Sosoe-2013} and references therein).

In the case of polynomial functions $f$ of order $k$ this is
equivalent to showing that the joint distribution of the first $k$
traces of Chebyshev polynomials $\{F_n(\H)\}_{n=1}^k$ converge in
distribution to independent Gaussian distributed random variables. We
are going to show that this result is also true of the UME, except
that, as noted in Section \ref{Trace formula}, we have $\Tr[\CY(\Phi)]
= \Tr[\CY(\Phi)^2] = 0$ and therefore $F_0(\CM) = F_1(\CM) = F_2(\CM)
\equiv 0$. For this reason it only makes sense to investigate
$F_n(\CM)$ for $n \geq 3$. One may consider this to be a substantial
deviation from the GUE, however it is shown in \cite{Maciazek-2015}
that one may always scale the first two moments in the Gaussian
$\beta$-ensembles in such a way that they may be considered
independently of all other moments.

For the UME (and more generally for Wigner matrices) the convergence
of $F_n(\CM)$ to independent Gaussians (see Theorem \ref{Convergence
  theorem} below) is well known (see e.g. \cite{Sodin-2016} for
instance) and can be obtained by showing the convergence of all
moments via combinatorial methods.
\begin{theorem}\label{Convergence theorem}
For $\CM$ distributed according to the UME and $k$
fixed\footnote{Although we are not aware of a specific result for the
  UME, it has been shown for real Wigner matrices by Sinai and
  Soshnikov \cite{Sinai-1998} that the traces $\Tr[\H^k]$ are still
  Gaussian distributed for growing $k$ (such that $k <
  N^{\frac{1}{2}}$) and one should expect a similar outcome here.}
\[
F(\CM) = (F_3(\CM),\ldots,F_k(\CM)) \to Z = (Z_3,\ldots,Z_k)
\]
in distribution as $N \to \infty$, where $Z_n$ are iid Gaussian
variables with mean zero and variance $\Av{Z_n^2} = \frac{n}{4}$.
\end{theorem}

However the combinatorial approach does not easily lend itself to
obtaining rates of convergence. To the best of our knowledge, the only
result that discusses rates of convergence for Wigner matrices
(including the UME) is by Chatterjee \cite{Chatterjee-2007}, in which
he opts instead for an analytical approach that combines Stein's
method with estimates of second order Poincar\'{e} inequalities. In
contrast, our result in Theorem \ref{Wasserstein distance estimation}
shows that one may still obtain similar convergence rates using
combinatorial procedures and, moreover, one only requires estimates on
finite moments of the matrix elements to achieve this.

\vspace{5pt}

In the present section we follow an Brownian motion approach pioneered
by Dyson in his seminal article \cite{Dyson-1962}, which outlined an
alternative description for the canonical Gaussian ensembles. His
insight was that one could replace the (static) Gaussian matrix
entries with independent (dynamical) Ornstein-Uhlenbeck processes such
that the equilibration in the large time limit corresponds to the
original Gaussian ensemble. The motion in the entries then induces a
corresponding stochastic motion in the eigenvalues, now known as Dyson
Brownian motion, whose stationary distribution is then given by the
joint probability density function for the eigenvalues of the Gaussian
ensemble.

In the setting of the UME (See Eq. (\ref{UME Definition})) we define
our motion such that each of the $N(N-1)/2$ phases $\phi_{\mu\nu}$ is
an independent standard Brownian motion on the torus $([0,2\pi))^{N
    (N-1) / 2}$. Thus, in a small time $\delta s$ each phase moves an
  amount $\delta \phi_{\mu\nu} := \phi_{\mu\nu}' - \phi_{\mu\nu}$,
  which is characterised by its drift and diffusion, given by the two
  respective moments
\begin{eqnarray}
\label{Phase drift} \E[\delta \phi_{\mu\nu} | \Phi] & := & \int d \Phi' (\phi'_{\mu\nu} - \phi_{\mu\nu}) \rho(\Phi \to \Phi';\delta s) =  \O(\delta s^2)  \\
\E[\delta \phi_{\mu\nu} \delta \phi_{\mu'\nu'} | \Phi] & := & \int d
\Phi' (\phi'_{\mu\nu} - \phi_{\mu\nu})(\phi'_{\mu'\nu'} -
\phi_{\mu'\nu'}) \rho(\Phi \to \Phi';\delta s) \nonumber \\
\label{Phase diffusion}  & & \ \ \ = 2(\delta_{\mu\mu'}\delta_{\nu\nu'} - \delta_{\mu\nu'}\delta_{\nu\mu'})\delta s + \O(\delta s^2)
\end{eqnarray}
and all higher moments are of order $\O(\delta s^2)$. Here $\rho(\Phi
\to \Phi';\delta s) = \prod_{\mu<\nu} \rho(\phi_{\mu\nu} \to
\phi_{\mu\nu}';s)$ denotes the probability of $\Phi =
\{\phi_{\mu\nu}\}_{\mu < \nu}$ moving to $\Phi' = \Phi + \delta \Phi$
in a time $\delta s$.

The above formulation of the motion is equivalent to considering the
probability distribution $P(\Phi;s)$ of finding the particles at
position $\Phi$ at time $s$, subject to some initial distribution
$P(\Phi;0)$ at time 0. $P(\Phi;s)$ satisfies the following
Fokker-Planck equation
\begin{equation}\label{Fokker-Planck equation}
\frac{\partial P(\Phi;s)}{\partial s} = \sum_{\mu < \nu} \frac{\partial^2 P(\Phi ; s)}{\partial \phi_{\mu\nu}^2}
\end{equation}
with periodic boundary conditions on the torus. Note that if
$P(\Phi;0) = \delta(\Phi - \Phi') = \prod_{\mu < \nu}
\delta(\phi_{\mu\nu} - \phi'_{\mu\nu})$, i.e. the particles are
conditioned to be at position $\Phi$ at time 0, then $P(\Phi';s) =
\rho(\Phi \to \Phi';s)$ is simply the transition probability. One may
solve this equation explicitly (see
e.g. \cite{Risken-1989,Gardiner-2009}), although we shall not need the
exact form of the solution here. In the limit of large times the
probability distribution satisfies
\[
\lim_{s \to \infty}P(\Phi ;s) = (2\pi)^{-N(N-1)/2},
\]
which corresponds to the stationary distribution of the UME - the
solution obtained once the left hand side of (\ref{Fokker-Planck
  equation}) is set to zero and appropriately normalised.

If the process is started from equilibrium, i.e. $P(\Phi ;0) =
(2\pi)^{-N(N-1)/2}$, then for all times $s$ we have $P(\Phi';s) =
P(\Phi ;0)$, where $\Phi$ and $\Phi'$ are related by $\rho(\Phi \to
\Phi';s)$. In the probability literature, the pair $(\Phi,\Phi')$ is
then termed an \emph{exchangeable pair}, since both $\Phi$ and $\Phi'$
have the same distribution (see e.g. \cite{Meckes-2009}).

\vspace{5pt}

With the same philosophy as used by Dyson \cite{Dyson-1962}, a change
$\delta \Phi$ in the matrix elements induces a corresponding change in
our traces of Chebyshev polynomials of $\delta F_n := F_n(\CM(\Phi'))
- F_n(\CM(\Phi))$, where $\Phi' = \Phi + \delta \Phi$, in time $\delta
s$. The key difference in the UME in comparison to the GUE, is that
the motion is not invariant under unitary transformations\footnote{See
  \cite{Joyner-2015} for further discussion on this point in the case
  of Bernoulli matrices.} and consequently the evolution of $F_n(\CM)$
is not closed (i.e. cannot be described entirely in terms of the
$F_n(\CM)$ themselves). Nevertheless we can estimate the remainder
terms using the combinatorial procedures outlined in the previous
section.

For large $N$ we will show in the following subsections that the drift
and diffusion coefficients associated to the motion of $F_n(\CM)$
approximately (in the probabilistic sense of Theorem \ref{Steins
  method theorem} and Theorem \ref{Wasserstein distance estimation}
below) satisfy
\begin{equation}\label{Drift hope}
\fl \lim_{\delta s \to 0}\frac{\E[\delta F_n |\Phi]}{\delta s} : = \lim_{\delta s \to 0}\frac{1}{\delta s}\int d\Phi (F_n(\CM') - F_n(\CM)) \rho(\Phi \to \Phi'; \delta s) \approx - nF_n(\CM)
\end{equation}
\begin{equation}\label{Diffusion hope}
\fl \lim_{\delta s \to 0}\frac{\E[\delta F_n^2 |\Phi]}{\delta s} : = \lim_{\delta s \to 0}\frac{1}{\delta s}\int d\Phi'(F_n(\CM') - F_n(\CM))^2 \rho(\Phi \to \Phi'; \delta s) \approx 2 \left(\frac{n^2}{4}\right),
\end{equation}
with cross terms $\lim_{\delta s \to 0}\E[\delta F_n \delta F_m |\Phi]
/ \delta s \approx 0$ and higher terms identically zero. This means
the process behaves in a similar way to an Ornstein-Uhlenbeck (OU)
process generated by the following Fokker-Planck equation
\begin{equation}\label{OU FP equation}
\frac{\partial Q(X;s)}{\partial s} = \sum_{n=3}^k \left[n\frac{\partial (X_n Q(X;s))}{\partial X_n} + \frac{n^2}{4} \frac{\partial^2 Q(X;s)}{\partial X_n^2}\right],
\end{equation}
where $X = (X_3,\ldots,X_n)$. Moreover, since $\lim_{s \to \infty}
P(\Phi;s)= (2\pi)^{-N(N-1)/2}$ is the stationary distribution of the
UME, the associated stationary distribution of $F(\CM)$ will, in turn,
be approximately equal to the stationary solution of (\ref{OU FP
  equation}) given by
\[
Q(X): = \lim_{s \to \infty} Q(X;s) = \prod_{n=3}^k \sqrt{\frac{2}{\pi n}} \exp\left\{-\frac{2 X^2}{n}\right\},
\]

Theorem \ref{Steins method theorem} below (due to [Meckes
  \cite{Meckes-2009}) makes
  this notion precise and for reasons of brevity we shall not repeat
  their proof here. However we would like to highlight that the result
  follows from the Taylor expansion of an appropriate observable
  $h(F(\CM))$, with $\delta h := h(F(\CM')) - h(F(\CM))$
\[
\E[\delta h|\Phi] =  \sum_{n=3}^k \E[\delta F_n | \Phi ] \frac{\partial h}{\partial F_n} + \frac{1}{2} \sum_{n,m=3}^k  \E[\delta F_n \delta F_m | \Phi ] \frac{\partial^2 h}{\partial F_n \partial F_m} + \ldots \ .
\]
Dividing through by a factor $\delta s$ and then taking the limit
$\delta s \to 0$ we see that from (\ref{Drift hope}) and
(\ref{Diffusion hope}) the combination of the first two terms are
(approximately) equal to $\A h(F(\CM))$, where
\[
\A   := \sum_{n=3}^k \frac{n^2}{4} \frac{\partial^2}{\partial x_n^2} - nx_n\frac{\partial }{\partial x_n}.
\]
Notice this is the adjoint operator to that in the right hand side of
(\ref{OU FP equation}). Stein's Lemma ensures that if for any
twice-differentiable test function $h$ we have $\Av{\A h(Z)}= 0$ then
$Z = (Z_3,\ldots,Z_k)$ must be a multi-dimensional Gaussian random
variable with $\Av{Z_n}=0$ and $\Av{Z_nZ_m} = n\delta_{nm}/4$. In our
case the average over the UME is not identically 0 but $\Av{\A
  h(F(\CM)} \approx 0$, which suggests that $F(\CM)$ is approximately
Gaussian. Stein's method (originally developed in order to provide an
alternative proof of the CLT \cite{Stein-1972}) then allows one to
estimate the distance between $F(\CM)$ and the Gaussian variable $Z$
in a suitable metric by bounding the variance of $\A h(F(\CM)$.

\begin{dfn}[Wasserstein distance]Let us denote $\mathcal{L} := \{f : \R^k \to \R : |f(x) - f(y)| \leq \| x-y \| \}$ to be the set of all Lipschitz continuous functions and $X,Y$ be two $k$-dimensional random variables, then the Wasserstein distance between $X$ and $Y$ is
\[
d_{\rm W}(X,Y) := \sup_{f \in \mathcal{L}}|\Av{f(X)} - \Av{f(Y)}|.
\]
\end{dfn}
The Wasserstein distance provides a particular way to measure the
distance between two probability distributions and it often emerges as
a natural distance when utilising Stein's method. Moreover, if one has
sequence of random variables $X_N$ in which $d_{\rm W}(X_N,Y) \to 0$
then this implies that $X_N \to Y$ in distribution.

The following theorem, utilising Stein's method, is due to Meckes. Note that their results are more general than the following statement but we adapt it to our setting for purposes of clarity.

\begin{theorem}[Meckes \cite{Meckes-2009}]\label{Steins method theorem} Let $\CM$ and $\CM'$ be two random matrices with the same probability distribution (they are an exchangeable pair) and related via some transition probability $\rho(\CM \to \CM' ;s)$. Let $X\equiv X(\CM) =(X_3(\CM),\ldots,X_k(\CM))$ and $X' \equiv X(\CM') = (X_3(\CM),\ldots,X_k(\CM'))$ be two $k-2$ dimensional random variables dependent of $\CM$ and $\CM'$.  If
\begin{enumerate}
\item \label{Thm Drift}
\[
 \lim_{\delta s \to 0}\frac{\E[\delta X_n |\CM]}{\delta s} = - nX_n(\CM) + R_n(\CM) \ ,
\]
\item \label{Thm Diffusion}
\[
 \lim_{\delta s \to 0}\frac{\E[\delta X_n \delta X_m |\CM]}{\delta s} = \frac{n^2}{2}\delta_{nm} + R_{nm}(\CM) \ ,
\]
\item \label{Thm remainder}
\[
 \lim_{\delta s \to 0}\frac{\E[ |\delta X_n \delta X_m \delta X_l | |\CM]}{\delta s} = 0 \ ,
\]
\end{enumerate}
for all $n,m,l=3,\ldots,k$ and $R_n(\CM)$ and $R_{nm}(\CM)$ are
(potentially) random variables depending on $\CM$, then
\begin{equation}\label{Thm final result}
d_{\rm W}(X,Z) \leq \frac{1}{3}\sum_{n=3}^k \Av{|R_n(\CM)|} + \frac{1}{9}\sqrt{\frac{2}{\pi}}\sum_{n,m=3}^k\Av{|R_{nm}(\CM)|},
\end{equation}
with $Z$ the multi-dimensional Gaussian random variable stated in Theorem \ref{Convergence theorem}.
\end{theorem}

We also comment that an alternative version of this theorem by D\"{o}bler and Stolz \cite{Dobler-2011} has also been used by Webb \cite{Webb-2015} to estimate the Wasserstein distance between the traces in the circular $\beta$-ensembles and Gaussian random variables.

\begin{theorem}\label{Wasserstein distance estimation}Let $\CM$ be distributed according to the UME (see Equation (\ref{1})) and let $F(\CM) = (F_3(\CM),\ldots,F_k(\CM))$ be defined as in (\ref{Centered Chebyshev}). Then, for $k$ fixed, we have
\[
d_{\rm W}(F(\CM),Z) = \O(N^{-1/2}),
\]
with $Z$ that of Theorem \ref{Convergence theorem}.
\end{theorem}

\begin{proof} In the following subsections we will show the remainders for our drift (\ref{Drift hope}) and diffusion (\ref{Diffusion hope}) terms satisfy $\Av{|R_n(\CM)|} = \O(N^{-1})$ and $\Av{|R_{nm}(\CM)|} = \O(N^{-\frac{1}{2}})$ respectively. Incorporating these estimates into the Wasserstein distance (\ref{Thm final result}) in Theorem \ref{Steins method theorem} then gives the result.
\end{proof}

We remark that a convergence rate of order $\O(N^{-\frac{1}{2}})$ is also found in \cite{Chatterjee-2007} for certain classes of Wigner matrices using the total-variation metric.

\subsection{Drift term}
Before proceeding, let us first introduce the following notation. Let $E := \{e = (\mu,\nu) : \mu < \nu\}$ be the set of directed edges on our graph such that $e(\sigma) = (\mu,\nu)$ if $\sigma=+$ and $(\nu,\mu)$ if $\sigma = -$. Then for a non-backtracking walk $w = (e_1(\sigma_1),e_2(\sigma_2),\ldots,e_n(\sigma_n))$ the total phase can be written as
\[
\Phi_w = \sum_{e \in w} \kappa^{(w)}_e \phi_e,
\]
where $\kappa^{(w)}_e = \#\{e(+) \in w\} - \#\{e(-) \in w\}$ counts the net number of traversals of the edge $e$ by $w$ (it may be positive, negative or zero). This means, using the properties of the motion (\ref{Phase drift}) and (\ref{Phase diffusion}), we have
\[
\E[\delta \Phi_w^2| \Phi] = \sum_{e,e' \in w} \kappa^{(w)}_e\kappa^{(w)}_{e'} \E[\delta \phi_e \delta \phi_{e'}|\Phi] = 2 \sum_{e \in w} (\kappa^{(w)}_e)^2 \delta s + \O(\delta s^2).
\]
Therefore, using the form of $F_n(\CM)$ from (\ref{Centered Chebyshev}) we have
\begin{eqnarray}
\fl \E[\delta F_n | \Phi] & = & \frac{1}{2(N-2)^{n/2}}\sum_{w \in \Lambda_n} \E[\exp\{i(\Phi_w + \delta \Phi_w)\} - \exp\{i\Phi_w\}|\Phi] \nonumber \\
\fl & = &  \frac{1}{2(N-2)^{n/2}}\sum_{w \in \Lambda_n} \E\left[\exp\{i\Phi_w\}(1 + i\delta \Phi_w - \frac{1}{2} \delta \Phi_w^2 + \ldots) - \exp\{i\Phi_w\} \bigg|\Phi\right]\nonumber \\
\fl & = &  -\frac{1}{2(N-2)^{n/2}}\sum_{w \in \Lambda_n} \exp\{i\Phi_w\} \frac{1}{2}\E[\delta \Phi_w^2 |\Phi] + \O(\delta s^2) \nonumber \\
\fl & = &  (- nF_n(\CM) + R_n(\CM) )\delta s + \O(\delta s^2),
\end{eqnarray}
where, writing $x_w := \sum_{e \in w}(\kappa^{(w)}_e)^2$ for simplicity, the remainder is given by
\[
\fl R_n(\CM) = \frac{1}{2(N-2)^{n/2}}\sum_{w \in \Lambda_n} \exp\{i\Phi_w\}(x_w -  n) = \frac{1}{2(N-2)^{n/2}}\sum_{w \in \Lambda'_n} \exp\{i\Phi_w\}x_w,
\]
and $\Lambda'_n : = \{ w \in \Lambda_n : x_w \neq n\}$. In particular, this excludes those $n$-periodic non-backtracking walks in which every edge is only traversed once. To estimate the value of this remainder we must compute
\begin{equation}\label{Drift remainder 1}
\Av{|R_n(\CM)|} \leq \sqrt{\Av{R_n(\CM)^2}} = \frac{1}{2}\sqrt{\sum_{w,w' \in \Lambda'_n} x_w x_{w'} \frac{\Av{\exp\{i(\Phi_w - \Phi_w')\}}}{(N-2)^n}}.
\end{equation}
By averaging over the phases we find the main contributions to (\ref{Drift remainder 1}) will come from pairs of walks in which $w = w'$. These, however, cannot be walks in which all the edges are distinct (as this would imply $x_w = n$). Therefore, the main contribution if from those $w$ containing precisely one edge that is traversed twice (once in each direction) and the remaining edges are connected to this edge by two loops in which every edge is traversed once. Taking $w'$ to be the same walk but in the opposite direction means we have, using the notation from Section \ref{umeexp}, $v = n - \beta +1 = n - 2$ vertices (note that we have effectively $\beta = 3$ loops since traversing an edge twice can be viewed as creating an additional loop). Following the arguments of Section \ref{umeexp} and Section \ref{Correlations of yn} this means $\sum_{w,w' \in \Lambda'_n} x_w x_{w'} \Av{\exp\{i(\Phi_w - \Phi_w')\}} =\O(N^{n-2})$ and thus $\Av{|R_n(\CM)|} = \O(N^{-1})$.

\subsection{Diffusion term}
We now show that the remainder for our diffusion term satisfies $\Av{|R_{nm}(\CM)|} = \O(N^{-1/2})$. To begin, using the notations above, we note that
\[
\fl
\E[\delta \Phi_w \delta \Phi_{w'} | \Phi] = \sum_{e \in w} \sum_{e' \in w'} \kappa^{(w)}_e \kappa^{(w')}_{e'} \E[\delta \phi_e \delta \phi_{e'}|\Phi] = 2\sum_e \kappa^{(w)}_e \kappa^{(w')}_{e} \delta s + \O(\delta s^2).
\]
Thus, writing $x_{w,w'} = \sum_e \kappa^{(w)}_e \kappa^{(w')}_{e}$ and taking the complex conjugate of $F_n(\CM)$ in the following (since it is real), we have for the diffusion term
\begin{eqnarray}
\fl \E[\delta F_n \delta F_m | \Phi] & = & \frac{1}{4(N-2)^{(n+m)/2}}\sum_{w \in \Lambda_n}\sum_{w' \in \Lambda_m} \exp\{i(\Phi_w - \Phi_{w'})\} \E[\delta \Phi_w \delta \Phi_{w'} |\Phi] \nonumber \\
\fl & =& \left(\frac{n^2}{2}\delta_{nm} + R_{nm}(\CM)\right)\delta s + \O(\delta s^2),
\end{eqnarray}
where
\begin{equation}\label{Diffusion remainder 1}
\fl R_{nm}(\CM) = \frac{1}{2(N-2)^{(n+m)/2}}\sum_{w \in \Lambda_n}\sum_{w' \in \Lambda_m} \exp\{i(\Phi_w - \Phi_{w'})\} x_{w,w'} - \frac{nm}{2} \delta_{nm}.
\end{equation}
In order to obtain an estimate for $\Av{|R_{nm}(\CM)|}$ we treat the cases $n = m$ and $n\neq m$ separately.

For $n=m$ we have
\[
\Av{|R_{nn}(\CM)|} \leq \frac{1}{2}\sqrt{\Av{\left(\sum_{w,w' \in \Lambda_n} \frac{\exp\{i(\Phi_w - \Phi_{w'})\} x_{w,w'}}{(N-2)^n} - n^2\right)^2}}
\]
Expanding out the brackets inside the square root above gives
\[
\fl
\sum_{w_1,w_2,w_3,w_4 \in \Lambda_n} \frac{ \Av{\exp\{i(\Phi_{w_1} - \Phi_{w_2} + \Phi_{w_3} - \Phi_{w_4})\}} x_{w_1,w_2}x_{w_3,w_4}}{(N-2)^{2n}}\]
\[
 - 2 n^2\sum_{w_1,w_2 \in \Lambda_n} \frac{ \Av{\exp\{i(\Phi_{w_1} - \Phi_{w_2})\}} x_{w_1,w_2}}{(N-2)^{n}} + n^4.
\]
Now, we first note that $x_{w,w'} = 0$ if $w$ and $w'$ do not share an
edge. Therefore the main contribution to the first sum in the above
comes from pairs of non-backtracking walks $w_1 = w_2$ and $w_3 = w_4$
that reside on disconnect subgraphs comprised of a single loop. As was
determined in (\ref{beta 1 contribution}), and using that $x_{w_1,w_1}
= n$ for such walks, means the first summation gives a contribution
$n^4 + \O(N^{-1})$ and the second summation $n^2 +
\O(N^{-1})$. Therefore, taking the square root, we have $\Av{|R_{nn}(\CM)|} = \O(N^{-1/2})$.

For $n \neq m$ we have
\[
\fl
\Av{|R_{nm}(\CM)|} \leq \sqrt{ \sum_{w_1,w_3 \in \Lambda_n}\sum_{w_2,w_4 \in \Lambda_m} \frac{ \Av{\exp\{i(\Phi_{w_1} - \Phi_{w_2} + \Phi_{w_3} - \Phi_{w_4})\}} x_{w_1,w_2}x_{w_3,w_4}}{(N-2)^{(n+m)}}}.
\]
Due to the presence of $x_{w_1,w_2}$ and $x_{w_3,w_4}$ it must be the case that for $\Av{\exp\{i(\Phi_{w_1} - \Phi_{w_2} + \Phi_{w_3} - \Phi_{w_4})\}} x_{w_1,w_2}x_{w_3,w_4}$ to be non-zero $w_1$ and $w_2$ must share at least one edge, and similarly for $w_3$ and $w_4$. Therefore, since $w_1$ and $w_2$ are of different lengths, those subgraphs supporting the walks $(w_1,\ldots,w_4)$ in which every edge is traversed twice must have at least three loops. The number of vertices in such a situation is therefore $v = n+m - 2$ and hence the contribution inside the square root is of order $\O(N^{-2})$, meaning $\Av{|R_{nm}(\CM)|} = \O(N^{-1})$.

\subsection{Remainder term}
In order to complete the necessary conditions for Theorem \ref{Wasserstein distance estimation} requires verifying that our motion satisfies Part (\ref{Thm remainder}) in Theorem \ref{Convergence theorem}. This is easily shown, since
\[
 \E[|\delta F_n \delta F_m\delta F_l | |\CM] \leq \sqrt{\E[(\delta F_n \delta F_m\delta F_l )^2 |\CM]} = \O(\delta s^{3/2}).
\]
This comes from noting that the main contribution to $\E[(\delta F_n
  \delta F_m\delta F_l )^2 |\CM]$ will come from terms in which the
same edges appears six times. Therefore, taking this conditional
expectation we find $\E[\delta \phi_e^6|\Phi] = \O(\delta s^3)$,
because $\delta \phi_e$ behaves like a Gaussian in the small time limit with
variance $\sigma^2 \propto \delta s$. Thus, dividing through by
$\delta s$ and taking the limit we get the result
\[
\lim_{\delta s \to 0} \frac{\E[|\delta F_n \delta F_m\delta F_l | |\CM]}{\delta s} = 0,
\]
as required.

\section{Conclusions}

In this paper we have studied the spectral density and moments of the
UME for large matrix dimension $N$. The UME is a paradigmatic example
of a Wigner ensemble where the distribution of the matrix elements is
not Gaussian. The spectral properties of the UME agree with those of
the GUE in the limit $N \to \infty$. We have focused attention on
contributions of next order(s) in an asymptotic expansion in $1 / N$.
These are of interest in their own right, possibly showing deviations
from universality. Their study also casts light on the power of the
approaches used for their study. We have used three very different
approaches: The supersymmetry approach, the graph-theoretical
approach, and the Brownian-motion approach.

Using the supersymmetry approach and the saddle-point approximation we
have shown that the leading $1 / N$ corrections to the spectral
density of the UME are the same as for the GUE and account for the
oscillations that feature so prominently in the numerical data. We
have not been able to push the supersymmetry approach beyond that
point. The graph-theoretical approach has yielded leading $1 / N$
corrections to the mean values of Chebyshev moments and their
covariances. These differ from GUE values. That information did not
suffice, however, to account for the oscillations in the spectral
density. Progress might be possible via more complicated combinatorial
computations. These would go beyond the scope of the present
paper. Finally we used a Brownian motion approach combined with
Stein's method to analyse the fluctuations of the spectral moments. In
particular the combinatorial ideas outlined in Section \ref{Graph
  theory approach} allowed us to estimate the error terms and provide
bounds on the rate of convergence to a multi-dimensional Gaussian in
the large $N$ limit.

One issue worth mentioning, that has not been addressed here, is that
due to gauge invariance the spectrum of any matrix in the UME depends
only on the magnetic fluxes $\Phi_c = \sum_{i} \phi_{v_i,v_{i+1}}$ on
the fundamental cycles $c$. The number of such fundamental cycles is
$\beta = \frac{1}{2}N(N-1) - N +1$, which is less than the number of
independent phases by $N-1$. We did not attempt to make use of the
gauge invariance or the freedom in choosing the independent cycles.

\section*{Acknowledgements}
CHJ and US would like to thank Professor Sasha Sodin for many
discussions and illuminating comments. CHJ is also grateful to the
Leverhulme Trust (ECF-2014-448) for financial support.

\end{document}